\documentclass[a4paper]{article}
\usepackage{pdfsync}
\usepackage[latin1]{inputenc}
\usepackage{graphicx}
\usepackage{amssymb,amsfonts,amsmath}
\usepackage{theorem}
\usepackage{hyperref}
\usepackage{graphicx}
\usepackage{multirow}
\usepackage{verbatim}
\usepackage{cite}

\newtheorem{proposition}{Proposition}[section]
{\theorembodyfont{\rmfamily}

\newtheorem{remark}{Remark}[section]
\newtheorem{example}{Example}[section]
}
\newenvironment{proof}[1][Proof]{\noindent\textbf{#1.} }{\newline \hspace*{\textwidth}\hspace*{-0,4cm} \rule{0.5em}{0.5em} \vspace{0,2cm}}

\begin{document}

\title{\textbf{An algorithm for computing geometric relative velocities through Fermi and observational coordinates}}


\author{Vicente J. Bol\'os
\\{\small Dpto. Matem\'aticas para la Econom\'{\i}a y la Empresa, Facultad de Econom\'{\i}a,}\\ {\small Universidad de Valencia. Avda. Tarongers s/n. 46022, Valencia, Spain.}\\ {\small e-mail\textup{: \texttt{vicente.bolos@uv.es}}}}


\maketitle

\begin{abstract}
We present a numerical method for computing the \textit{Fermi} and \textit{observational coordinates} of a distant test particle with respect to an observer. We apply this method for computing some previously introduced concepts of relative velocity: \textit{kinematic}, \textit{Fermi}, \textit{spectroscopic} and \textit{astrometric} relative velocities. We also extend these concepts to non-convex normal neighborhoods and we make some convergence tests, studying some fundamental examples in Schwarzschild and Kerr spacetimes. Finally, we show an alternative method for computing the Fermi and astrometric relative velocities.
\end{abstract}



\section{Introduction}

In Newtonian mechanics, the concept of ``relative velocity'' of a test particle with respect to an observer is unambiguous and fundamental. Nevertheless, in general relativity it is only well defined when the observer and the test particle are in the same event. In order to generalize this concept for distant test particles, a definition of relative velocity based on a particular coordinate system can be introduced, and it could be interesting in some cases, like the comoving coordinates in FLRW spacetimes for static observers (see \cite{Brae12}). Moreover, some notions of relative velocity of a distant test particle were introduced by the IAU using adapted reference systems in the case of objects in the neighborhood of the solar system (see \cite{Soff03,Lind03}). However, a universal concept of relative velocity independent from any coordinate system is basic and thereby, some authors have proposed geometric definitions without any coordinate-dependence (see \cite{Nar94,Bini95,Ca06}).
In this way, four different intrinsic geometric definitions of relative velocity of a distant test particle with respect to a single observer were introduced in \cite{Bol07}. These definitions are strongly associated with the concept of simultaneity: \textit{kinematic} and \textit{Fermi} in the framework of ``spacelike simultaneity'', \textit{spectroscopic} and \textit{astrometric} in the framework of ``lightlike simultaneity''.

These four concepts of relative velocity each have full physical sense, and have proved to be useful in the study and interpretation of properties of particular spacetimes (see \cite{Bol07,KC10,KR11,BK12}). For example, we can measure the expansion of space from the kinematic and Fermi relative velocities of comoving observers of FLRW spacetimes (see \cite{BHK12,Klein13}), or find the frequency shift and the light aberration effect (see \cite{Bol05}) from the spectroscopic relative velocity.

But, in most cases, the computations are analytically very complex and it makes the theoretical study much harder. These computations are strongly associated with the \textit{Fermi} and \textit{observational coordinates} and so, we present in this paper a numerical algorithm for finding these coordinates, allowing the computation of the geometric relative velocities.

This paper is organized as follows. In Section \ref{sec2} we present the framework, establishing the notation and defining some necessary concepts, introducing in Section \ref{sec2.1} the four geometric concepts of relative velocity and extending the original definitions to non-convex normal neighborhoods. In Section \ref{sec3} we develop the algorithm focused on computing the \textit{Fermi} and \textit{observational coordinates} of the test particle with respect to the observer by means of finding a geodesic (with certain characteristics) from the observer to the test particle, and we also make a discussion about the convergence of the method in Remark \ref{rem:conv}. In Section \ref{sec:examples} we give some fundamental examples in Schwarzschild and Kerr spacetimes, showing the rate of convergence and the effectiveness of the algorithm. Finally, in Appendix \ref{sec:alter} we present an alternative method for computing the Fermi and astrometric relative velocities.

\section{Definitions and notation}
\label{sec2}

We work in a Lorentzian spacetime manifold $\left( \mathcal{M},g\right) $, with $c=1$ and $\nabla $ the Levi-Civita connection, using the ``mostly plus'' signature convention $(-,+,+,+)$. Given two events $p$, $q$, and a segment curve $\psi $ that joins $p$ and $q$, the parallel transport from $p$ to $q$ along $\psi $ is denoted by $\tau ^{\psi } _{pq}$. Given a curve $\beta :I\rightarrow \mathcal{M}$ with $I\subseteq \mathbb{R}$, the image $\beta I$ (a subset in $\mathcal{M}$) is identified with $\beta $. Vector fields are denoted by uppercase letters and vectors (defined at a single point) are denoted by lowercase letters. If $u$ is a vector, then $u^{\bot }$ denotes the orthogonal space of $u$. The projection of a vector $v$ onto $u^{\bot }$ is the projection parallel to $u$, i.e. $v-\frac{g(u,v)}{g(u,u)}u$. Moreover, if $x$ is a spacelike vector, then $\Vert x\Vert :=g\left( x,x\right) ^{1/2}$ is the modulus of $x$. Given a vector field $X$, the unique vector of $X$ in $T_p\mathcal{M}$ is denoted by $X_p$.

In general, we say that a timelike world line $\beta $ is an \textit{observer} (or a \textit{test particle}); nevertheless, we say that a future-pointing timelike unit vector $u$ in $T_{p}\mathcal{M}$ is also an \textit{observer at $p$}, identifying the observer with its 4-velocity.

A \textit{light ray} is a lightlike (null) geodesic $\lambda $. A \textit{light ray from }$q$\textit{\ to }$p$ is a light ray $\lambda $ such that $q,p\in \lambda $ and $p$ is in the causal future of $q$.

We are going to consider two kinds of intrinsic simultaneity: \textit{spacelike} and \textit{lightlike}. Given an observer $u$ at $p$, the events simultaneous with $u$ form the corresponding \textit{simultaneity submanifold}:
\begin{itemize}
\item Spacelike simultaneity: the \textit{Fermi surface} $L_{p,u}$ (also known as \textit{Landau submanifold}) is given by all the geodesics starting from $p$ and orthogonal to $u$. In terms of the exponential map\footnote{Given $v\in T_p\mathcal{M}$, $\exp _p v:=\gamma _v(1)$ where $\gamma _v$ is the geodesic starting at $p$ with initial tangent vector $v$.} on $T_p\mathcal{M}$, it is given by $\exp _p u^{\bot}$.
\item Lightlike simultaneity: the \textit{past-pointing horismos submanifold} $E^-_p$ is given by all the light rays arriving at $p$, i.e. it is given by $\exp _p C_p^-$ where $C_p^-$ is the past-pointing light cone in $T_p\mathcal{M}$ composed by all the past-pointing lightlike vectors of $T_p\mathcal{M}$.
\end{itemize}

\subsection{Geometrically defined relative velocities}
\label{sec2.1}

Four different definitions of relative velocity of a test particle with respect to an observer were introduced in \cite{Bol07}, working in a convex normal neighborhood, where given two different events there exists a unique geodesic joining them. Now, we are going to work in a general spacetime $\mathcal{M}$, not necessarily a convex normal neighborhood, and we are going to extend the definitions of \cite{Bol07} to this new setting, where two events could be joined by more than one (or none) geodesic and the simultaneity submanifolds could present self-intersections.

Throughout the paper, we consider an observer $\beta $ and a test particle $\beta '$ (parameterized by their proper times) with 4-velocities $U$ and $U'$ respectively. Moreover, we  consider an event $p$ of $\beta $ with 4-velocity $u:=U_p$.

Given a vector $s\in u^{\bot }$ such that $\exp _p s\in \beta '$, the corresponding \textit{kinematic relative velocity of $\beta '$ with respect to $u$} is defined by the vector
\begin{equation}
\label{vkin}
v_{\mathrm{kin}}:=\frac{1}{-g\left( \tau ^{\psi }_{q_{\mathrm{s}}p}u'_{\mathrm{s}},u\right) }\tau ^{\psi }_{q_{\mathrm{s}}p}u'_{\mathrm{s}}-u,
\end{equation}
where $q_{\mathrm{s}}:=\exp _p s$ is the event of $\beta '$ at which the relative velocity is measured, $u'_{\mathrm{s}}:=U'_{q_{\mathrm{s}}}$ is the 4-velocity of $\beta '$ at $q_{\mathrm{s}}$, and $\tau ^{\psi }_{q_{\mathrm{s}}p}$ is the parallel transport from $q_{\mathrm{s}}$ to $p$ along the geodesic segment given by $\psi \left( \alpha \right) :=\exp _p \alpha s$ for $0\leq \alpha \leq 1$ (see Figure \ref{diagram2}, left). In this case, $s$ is a \textit{relative position of $\beta '$ with respect to $u$}, and there is a different $v_{\mathrm{kin}}$ for each different $s$ satisfying $s\in u^{\bot }$, $\exp _p s\in \beta '$. Note that if we work in a convex normal neighborhood then $s$ is unique.

\begin{figure}[tbp]
\begin{center}
\includegraphics[width=0.8\textwidth]{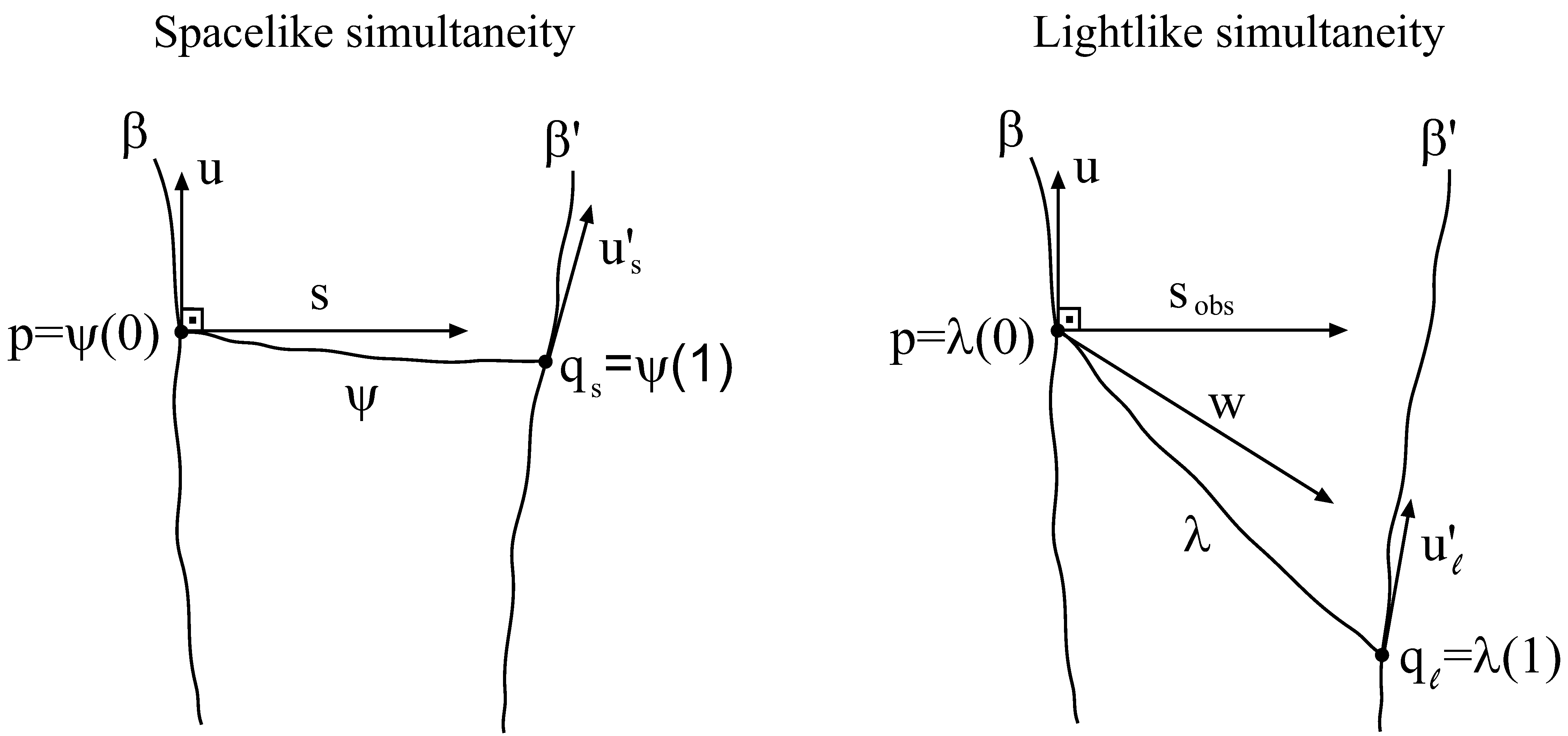}
\end{center}
\caption{Scheme of the elements involved in the study of relative velocities of $\beta '$ with respect to $u$ in general (not necessarily in a convex normal neighborhood). Left: given $s\in u^{\bot }$ such that $\exp _p s\in \beta '$, we define $\psi \left( \alpha \right) :=\exp _p \alpha s$ for $0\leq \alpha \leq 1$, $q_{\mathrm{s}}:=\psi (1)$, and $u'_{\mathrm{s}}$ is the 4-velocity of $\beta '$ at $q_{\mathrm{s}}$. Right: given $w\in C^-_p$ such that $\exp _p w\in \beta '$, we define $s_{\mathrm{obs}}:=w+g(u,w)u$, $\lambda \left( \alpha \right) :=\exp _p \alpha w$ for $0\leq \alpha \leq 1$, $q_{\ell}:=\lambda (1)$, and $u'_{\ell}$ is the 4-velocity of $\beta '$ at $q_{\ell}$.} \label{diagram2}
\end{figure}

Analogously, we can define another concept of relative velocity also introduced in \cite{Nar94}: given a vector $w\in C^-_p$ such that $\exp _p w\in \beta '$, the corresponding \textit{spectroscopic relative velocity of $\beta '$ with respect to $u$} is defined by the vector
\begin{equation}
\label{vspec}
v_{\mathrm{spec}}:=\frac{1}{-g\left( \tau ^{\lambda }_{q_{\ell}p}u'_{\ell},u\right) }\tau ^{\lambda }_{q_{\ell}p}u'_{\ell}-u.
\end{equation}
where $q_{\ell}:=\exp _p w$ is the event of $\beta '$ at which the relative velocity is measured, $u'_{\ell}:=U'_{q_{\ell}}$ is the 4-velocity of $\beta '$ at $q_{\ell}$, and $\tau ^{\lambda }_{q_{\ell}p}$ is the parallel transport from $q_{\ell}$ to $p$ along the light ray segment given by $\lambda \left( \alpha \right) :=\exp _p \alpha w$ for $0\leq \alpha \leq 1$ (see Figure \ref{diagram2}, right). In this case, the projection of $w$ onto $u^{\bot }$ given by $s_{\mathrm{obs}}:=w+g(u,w)u$ is the corresponding \textit{observed relative position of $\beta '$ with respect to $u$}, and there is a one-to-one correspondence between $w$ and $s_{\mathrm{obs}}$ (note that $w=s_{\mathrm{obs}}-\| s_{\mathrm{obs}}\| u$). So, there is a different $v_{\mathrm{spec}}$ for each different $w$ satisfying $w\in C^-_p$, $\exp _p w\in \beta '$, and this means that if there is gravitational lensing then each image of the observed object has a different spectroscopic relative velocity.

\begin{remark}
The spectroscopic relative velocity is specially useful because the frequency shift can be deduced from it (see \cite{Bol07}): let $\lambda $ be a light ray from $q_{\ell}$ to $p$ and let $u$, $u'_{\ell}$ be two observers at $p$, $q_{\ell}$ respectively; then
\begin{equation}
\label{dopplergen}
\frac{\nu '}{\nu }=\frac{1}{\sqrt{1-\|v_{\mathrm{spec}}\| ^2}}\left(1+g\left( v_{\mathrm{spec}},\frac{s_{\mathrm{obs}}}{\| s_{\mathrm{obs}}\|}\right) \right) ,
\end{equation}
where $\nu $, $\nu '$ are the frequencies of $\lambda $ observed by $u$, $u'_{\ell}$ respectively, $v_{\mathrm{spec}}$ is the corresponding spectroscopic relative velocity given by \eqref{vspec}, and $s_{\mathrm{obs}}$ is the corresponding observed relative position of $u'_{\ell}$ with respect to $u$.
\end{remark}

In general, the existence of $s$ or $w$ (and consequently $s_{\mathrm{obs}}$) is not assured. Nevertheless, if they exist for each event of the observer $\beta $, we can construct (differentiable) vector fields $S$ and $S_{\mathrm{obs}}$ defined on $\beta $, representing a \textit{relative position} and an \textit{observed relative position of $\beta '$ with respect to $\beta $}, respectively. But, of course, not all choices of $s$ and $w$ lead to differentiable $S$ and $S_{\mathrm{obs}}$, only those that change differentiably when varying the event of the observer. From $S$ and $S_{\mathrm{obs}}$, we can construct the vector fields $V_{\mathrm{kin}}$ and $V_{\mathrm{spec}}$ defined on $\beta $, representing a \textit{kinematic} and a \textit{spectroscopic relative velocity of $\beta '$ with respect to $\beta $}, respectively.

Given a vector field $S$ defined on $\beta $ and representing a relative position of $\beta '$ with respect to $\beta $ (i.e. such that $S_p\in U_p^{\bot }$ and $\exp _p S_p \in \beta '$ for all $p\in \beta $), the corresponding \textit{Fermi relative velocity of $\beta '$ with respect to $\beta $} is the vector field
\begin{equation}
\label{vfermi}
V_{\mathrm{Fermi}}:=\nabla _{U}S+g\left( \nabla _{U}S,U\right) U=\nabla _{U}S-g\left( S,\nabla _{U}U\right) U,
\end{equation}
defined on $\beta $.

Analogously, given a vector field $S_{\mathrm{obs}}$ defined on $\beta $ and representing an observed relative position of $\beta '$ with respect to $\beta $ (i.e. such that $S_{\mathrm{obs}~p}\in U_p^{\bot }$ and $\exp _p ( S_{\mathrm{obs}~p}-\| S_{\mathrm{obs}~p}\| U_p) \in \beta '$ for all $p\in \beta $), the corresponding \textit{astrometric relative velocity of $\beta '$ with respect to $\beta $} is the vector field
\begin{equation}
\label{vast}
V_{\mathrm{ast}}:=\nabla _{U}S_{\mathrm{obs}}+g\left( \nabla _{U}S_{\mathrm{obs}},U\right) U=\nabla _{U}S_{\mathrm{obs}}-g\left( S_{\mathrm{obs}},\nabla _{U}U\right) U,
\end{equation}
defined on $\beta $.

If we work in a convex normal neighborhood, then it is assured that there exists a unique $S$ and $S_{\mathrm{obs}}$, and hence there exists a unique $V_{\mathrm{kin}}$, $V_{\mathrm{spec}}$, $V_{\mathrm{Fermi}}$, and $V_{\mathrm{ast}}$ (see \cite{Bol07}). But this is not true in general and so different vector fields $S$ and $S_{\mathrm{obs}}$ define different vector fields $V_{\mathrm{kin}}$, $V_{\mathrm{spec}}$, $V_{\mathrm{Fermi}}$, and $V_{\mathrm{ast}}$.

In order to complete the notation that we are going to use, we define the vectors $v_{\mathrm{Fermi}}:=V_{\mathrm{Fermi}\, p}$ and $v_{\mathrm{ast}}:=V_{\mathrm{ast}\, p}$; moreover, throughout the paper we are going to denote $s:=S_p$, $s_{\mathrm{obs}}:=S_{\mathrm{obs}\, p}$, $v_{\mathrm{kin}}:=V_{\mathrm{kin}\, p}$, and $v_{\mathrm{spec}}:=V_{\mathrm{spec}\, p}$ as we have already done in this section.

\section{The algorithm}
\label{sec3}

First, we are going to suppose that we work in a convex normal neighborhood, and later, in Section \ref{nonconvex} we will extend the discussion to non-convex normal neighborhoods. So, working in a convex normal neighborhood implies that given an event $p\in \beta $ with 4-velocity $u$, there exists a unique event $q_{\mathrm{s}}\in \beta '$ such that the unique geodesic $\psi $ that joins $p$ and $q_{\mathrm{s}}$ is in $L_{p,u}$ (i.e. it is orthogonal to $u$ at $p$); on the other hand, there exists a unique event $q_{\ell}\in \beta '$ such that the unique geodesic $\lambda $ that joins $p$ and $q_{\ell}$ is in $E_p^-$ (i.e. it is a light ray arriving at $p$). In this case, the main difficulty is to find the geodesics $\psi $ (spacelike simultaneity) and $\lambda $ (lightlike simultaneity), taking into account that the events $q_\mathrm{s}$ and $q_{\ell}$ are also unknown (see Figure \ref{diagram2}). This is equivalent to finding the relative positions $s$ and $s_{\mathrm{obs}}$, i.e. the \textit{Fermi} and \textit{observational} (or \textit{optical}) coordinates respectively (see \cite{Ferm22,Walker35,Mana63,Elli80,Elli85}). Once we have found them, the corresponding relative velocities are easy to compute by means of their definitions (see Section \ref{sec2.1}).

If we can not find theoretically the Fermi and observational coordinates, then we can not find the geodesics $\psi $ and $\lambda $ of Figure \ref{diagram2}. Of course, there is a \textit{brute-force} method for finding these geodesics that consists on launching a lot of geodesics from the observer at $p$ in different directions (orthogonal to $u$ in the spacelike case, and past-pointing lightlike directions in the lightlike case), trying to reach the test particle $\beta '$. But obviously, this method spends a lot of computation time, and solving the geodesic equations with an acceptable accuracy is not as fast as we desire in some metrics.

So, we are going to propose a more efficient method based on an iterative correction algorithm with a Newton-Raphson structure. But first, we need to make more precise the concept of ``nearness'' between curves.

\subsection{The concept of ``nearness''}
\label{sec:nearness}

There is no global concept of distance in pseudo-Riemannian manifolds because the metric $g$ is degenerate. Nevertheless, given a 3-dimensional spacelike foliation in $\mathcal{M}$ we have that the induced metric $\overline{g}$ is Riemannian. Hence, we are going to suppose that we use a coordinate system $\left\{ x^0,x^1,x^2,x^3\right\} $ such that $x^1,x^2,x^3$ are spacelike coordinates and $x^0$ (also denoted as $t$) is a timelike coordinate, referred as \textit{coordinate time}; then, $t=\mathrm{constant}$ defines the leaves of the desired spacelike foliation, and the covariant coefficients of the induced Riemannian metric $\overline{g}$ are given by $\overline{g}_{ij}=g_{ij}$, where $g_{\mu \nu}$ are the covariant coefficients of the general metric $g$ in the above coordinate system (Latin indices run over $1,2,3$, and Greek indices run over $0,1,2,3$). Namely, given two vectors $v,w$ in the same tangent space of a leaf, we have that $\overline{g}(v,w)=g_{ij}v^i w^j$, and $\Vert v\Vert = \overline{g}(v,v)^{1/2}$.

Therefore, given any event $p$ in a leaf, there exists a local concept of spatial distance for events $q$ in the same leaf: $\textrm{d}(p,q):=\Vert v\Vert $, where $v$ is the vector which ``joins'' $p$ and $q$, i.e. such that $\overline{\exp }_p v=q$ where $\overline{\exp }_p$ is the induced exponential map at $p$. If $\textrm{d}(p,q)\approx 0$, the tangent spaces at $p$ and $q$ can be identified by means of the coordinate system, and we can consider an affine structure\footnote{Given an affine structure, we can subtract points to get vectors, or add a vector to a point to get another point.} around $p$, having $\overline{\exp }_p v\approx p+v$, i.e. $v^i \approx q^i-p^i$. So, we will say that two events $p,q$ with the same coordinate time are \textit{close} if $\Vert q-p\Vert $ is considered to be small, where $q-p$ is the vector in the tangent space of $p$ with coordinates $q^i-p^i$.

\begin{remark}
Note that the geodesics in the leaves $t=\mathrm{constant}$ are not the same as in the original manifold (unless the leaves submanifolds were totally geodesic), and so it is not assured the existence and uniqueness of the previous vector $v$ in general. Nevertheless, in this section we work in a convex normal neighborhood, i.e. the exponential map is a diffeomorphism; hence, if the leaves $t=\mathrm{constant}$ are regular submanifolds (i.e. the vector field $\frac{\partial}{\partial t}$ is synchronizable) then the induced exponential maps are also diffeomorphisms and so the leaves are convex normal neighborhoods. In conclusion, we have to add the assumption that $\frac{\partial}{\partial t}$ is synchronizable.
\end{remark}

This concept of nearness can be also applied to curves: we will say that two curves $c$, $c'$ are \textit{close} if there exist two events $p\in c$ and $q\in c'$ with the same coordinate time such that $p$ and $q$ are close (i.e. $\Vert q-p\Vert $ is small).

\subsection{Spacelike simultaneity}
\label{sec:3.2}

Let $\left\{ t\equiv x^0,x^1,x^2,x^3\right\} $ be a coordinate system such that $x^1,x^2,x^3$ are spacelike coordinates and $t$ is a timelike coordinate.
In the framework of spacelike simultaneity and taking into account the previous concept of nearness, we choose an initial vector $s_0\in u^{\bot }$, such that the geodesic $\psi _0$ starting from $p$ with initial tangent vector $s_0$ is sufficiently close to the test particle $\beta '$; i.e. there exist events $q_{\mathrm{geo}}$ (in the geodesic) and $q_{\mathrm{part}}$ (in the test particle), both with the same coordinate time, such that $\Vert q_{\mathrm{part}}-q_{\mathrm{geo}}\Vert $ is considered to be small.

\begin{remark}
\label{rem3.1}
In practice, we need a sub-algorithm for estimating $q_{\mathrm{geo}}$ and $q_{\mathrm{part}}$: given the initial geodesic $\psi _0$, we can compute the spatial distance between events of $\psi _0$ and the corresponding events of $\beta '$ with the same coordinate time; supposing that $q_{\mathrm{geo}}$ (in $\psi _0$) and $q_{\mathrm{part}}$ (in $\beta '$) are the events that minimize this distance, we can use a bisection method for estimating them with low computational cost. Nevertheless, the accuracy at this stage is not too much important because we only need two events sufficiently close.
\end{remark}

\begin{figure}[tbp]
\begin{center}
\includegraphics[width=0.75\textwidth]{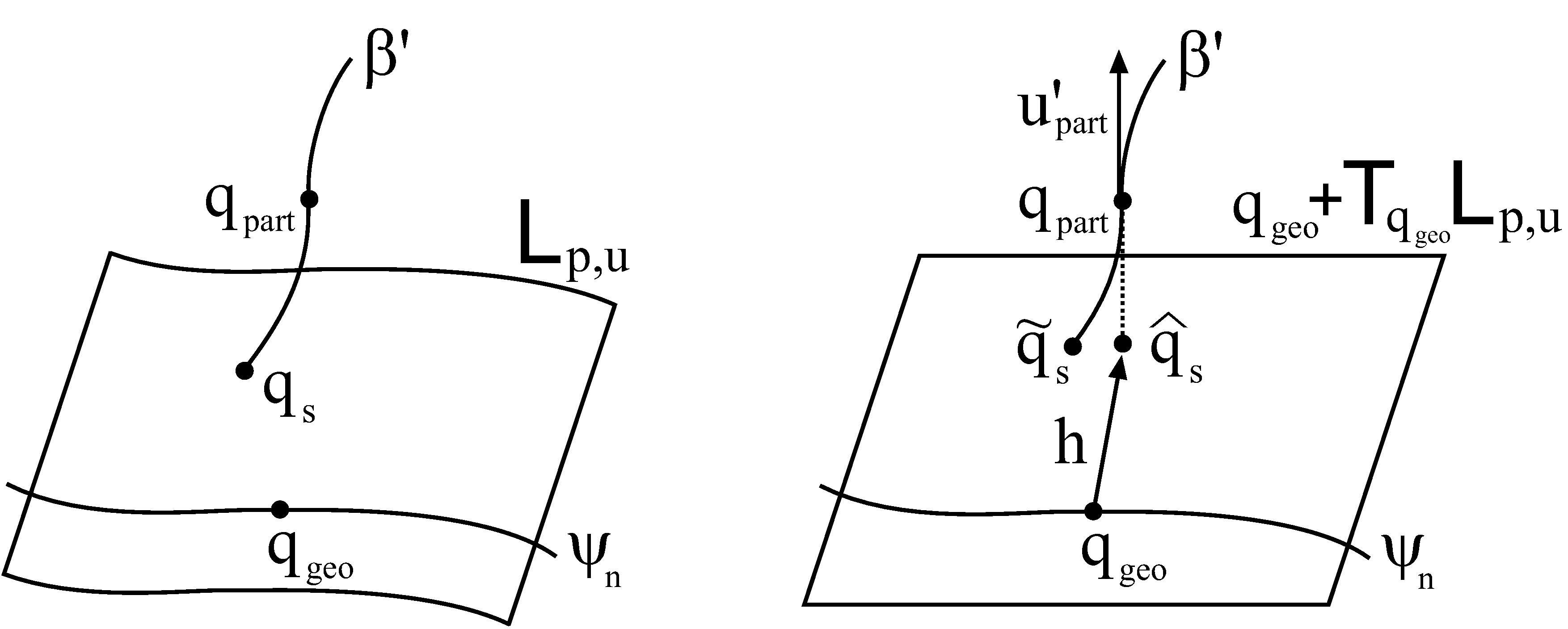}
\end{center}
\caption{Diagrams for $q_{\mathrm{s}}$ (left) and the estimations of $q_{\mathrm{s}}$ (right) when $q_{\mathrm{geo}}$ and $q_{\mathrm{part}}$ (that are events with the same coordinate time) are close. In this case, all the tangent spaces are identified by means of the coordinate system and provide an affine structure around $q_{\mathrm{geo}}$. Right: the Fermi surface $L_{p,u}$ is approximated by the affine hyperplane $q_{\mathrm{geo}}+T_{q_{\mathrm{geo}}}L_{p,u}$.} \label{vdr-sp}
\end{figure}

Since we are going to work near the event $q_{\mathrm{geo}}$, we can identify all the tangent spaces by means of the coordinate system and provide an affine structure in the vicinity of $q_{\mathrm{geo}}$, where $\exp _{q_{\mathrm{geo}}} v\approx q_{\mathrm{geo}}+v$. Hence, the Fermi surface $L_{p,u}$ can be linearly approximated by the affine hyperplane $q_{\mathrm{geo}}+T_{q_{\mathrm{geo}}}L_{p,u}$ and the intersection event $\widetilde{q}_{\mathrm{s}}$ between $\beta '$ and $q_{\mathrm{geo}}+T_{q_{\mathrm{geo}}}L_{p,u}$ is approximated by
\begin{equation}
\label{eq:qsapprox}
\widehat{q}_{\mathrm{s}}:= q_{\mathrm{geo}}+h,
\end{equation}
where $h$ is the projection of  $q_{\mathrm{part}}-q_{\mathrm{geo}}$ onto $T_{q_{\mathrm{geo}}}L_{p,u}$ parallel to $u'_{\mathrm{part}}$, with $u'_{\mathrm{part}}$ the $4$-velocity of $\beta '$ at $q_{\mathrm{part}}$ (see Figure \ref{vdr-sp} with $n=0$).

For finding $h$ we need first to find $T_{q_{\mathrm{geo}}}L_{p,u}$, and for this purpose we can use the Jacobian matrix $J^{\mu}_{\nu}:=\partial _{\nu} \exp _p ^{\mu} s_0$ where $\partial _{\nu}$ is the partial derivative with respect to the $\nu$-th coordinate (note that the derivatives of $\exp _p$ are easy to estimate numerically). So, $\left\{ \bar{e}_1,\bar{e}_2,\bar{e}_3\right\} $ is a basis of $T_{q_{\mathrm{geo}}}L_{p,u}$, where $\bar{e}_i^{\mu}:=J^{\mu}_{\nu}e^{\nu}_i$ and $\left\{ e_1,e_2,e_3\right\} $ is a basis of $u^{\bot}$, e.g. the projections onto $u^{\bot}$ of the spacelike vectors $\left. \frac{\partial}{\partial x^1}\right| _p\equiv (0,1,0,0)$, $\left. \frac{\partial}{\partial x^2}\right| _p\equiv (0,0,1,0)$ and $\left. \frac{\partial}{\partial x^3}\right| _p\equiv (0,0,0,1)$. Then, supposing that the Fermi surface $L_{p,u}$ is spacelike at $q_{\mathrm{geo}}$, we have that $\left\{ \bar{e}_1,\bar{e}_2,\bar{e}_3,u'_{\mathrm{part}}\right\} $ is a basis of $T_{q_{\mathrm{geo}}}\mathcal{M}$, and hence $h=\bar{\alpha}^i\bar{e}_i$, where $\bar{\alpha}^i$ are the spatial coordinates of $q_{\mathrm{part}}-q_{\mathrm{geo}}$ in the above basis.

\begin{remark}
Given $q\in L_{p,u}$, in \cite[Proposition 3]{Bol02} it is proved that, in some cases, $T_q L_{p,u}=\left( \tau _{pq}u\right) ^{\bot}$ and consequently, $L_{p,u}$ is spacelike at $q$; for example, in the case of \textit{stationary observers} in the Schwarzschild spacetime (see Examples \ref{example1} and \ref{example3}). Nevertheless, $T_q L_{p,u}$ is not $\left( \tau _{pq}u\right) ^{\bot}$ in general, as it occurs in the analogous case of \textit{stationary observers} in the Kerr spacetime (see Examples \ref{example2} and \ref{example4}), although it can be checked that the Fermi surface $L_{p,u}$ remains spacelike.
\end{remark}

The next step is finding $\widetilde{q}_{\mathrm{s}}$ by means of a Newton-Raphson method: we find $\widehat{q}_{\mathrm{s}}$ by \eqref{eq:qsapprox}, and we redefine $q_{\mathrm{part}}$ as the event of $\beta '$ with the same coordinate time as $\widehat{q}_{\mathrm{s}}$; then, by \eqref{eq:qsapprox} again, we find another $\widehat{q}_{\mathrm{s}}$, repeating this process until $\widehat{q}_{\mathrm{s}}$ approximates $\widetilde{q}_{\mathrm{s}}$ with the desired accuracy. The convergence of this method is assured by the next proposition.

\begin{proposition}
\label{prop:conv1}
If the Fermi surface $L_{p,u}$ is spacelike at $q_{\mathrm{geo}}$ and the acceleration of the test particle $\beta '$ is bounded, then the Newton-Raphson method described above produces a sequence of events $\left\{ \widehat{q}_{\mathrm{s}\,n}\right\} _{n\in \mathbb{N}}$ that converges to $\widetilde{q}_{\mathrm{s}}$ with quadratic order for a sufficiently close events $q_{\mathrm{part}}$ and $q_{\mathrm{geo}}$.
\end{proposition}
\begin{proof}
We are going to denote $\widehat{q}_{\mathrm{s}\,n}$ as $q_n$ for the sake of simplicity. Working in coordinates $\left\{ t,x^1,x^2,x^3\right\} $, let $\beta '$ be parameterized by the coordinate time $t$ and let $t_0:= q_{\mathrm{part}}^t$, $t_n:=q_{n-1}^t$ for $n\geq 1$.

Given $n\in \mathbb{N}$, the event $q_n$ is the intersection of the affine hyperplane $q_{\mathrm{geo}}+T_{q_{\mathrm{geo}}}L_{p,u}$ and the affine line $\beta '(t_n)+\langle \dot{\beta }'(t_n)\rangle $, where the overdot denotes derivation with respect to $t$ and ``$\langle \,\,\rangle$'' denotes the span. If $t+a_i x^i=b$ is the cartesian equation of $q_{\mathrm{geo}}+T_{q_{\mathrm{geo}}}L_{p,u}$ (we can suppose that the coefficient of $t$ is $1$ because $L_{p,u}$ is not timelike), then, by algebraic manipulations, we have that the coordinate time of $q_n$ is given by
\begin{equation}
\label{eq:prop1}
t_{n+1}=\frac{a_i \left( t_n \dot{x}^i_n-x^i_n\right) +b}{1+a_i \dot{x}^i_n},
\end{equation}
where $x^i_n:=\beta '^i(t_n)$ and $\dot{x}^i_n:=\dot{\beta }'^i(t_n)$ for $i=1,2,3$. It can be proved that $1+a_i \dot{x}^i_n\neq 0$ because $\beta '$ is timelike and $L_{p,u}$ is spacelike at $q_{geo}$, and so the intersection event $q_n$ always exists.

On the other hand, if $\widetilde{q}_{\mathrm{s}}=\beta '(\tilde{t})$, from the Taylor expansion of order $1$ of $\beta '^i(t)$ at $t_n$ we have that
\begin{equation}
\label{eq:prop2}
\widetilde{q}_{\mathrm{s}}\,\!\!^i=x^i_n+\dot{x}^i_n \epsilon _n + R^i_1,
\end{equation}
where $\epsilon _n:=\tilde{t}-t_n$ and $R^i_1:=\frac{1}{2}\ddot{\beta }'^i(\xi ^i_n)\epsilon _n^2$ with $\xi ^i_n$ between $\tilde{t}$ and $t_n$ for $i=1,2,3$. Moreover, since $\widetilde{q}_{\mathrm{s}}\in q_{\mathrm{geo}}+T_{q_{\mathrm{geo}}}L_{p,u}$, we have that
\begin{equation}
\label{eq:prop3}
\tilde{t}+a_i \widetilde{q}_{\mathrm{s}}\,\!\!^i=b.
\end{equation}
Substituting \eqref{eq:prop2} in \eqref{eq:prop3} and taking into account \eqref{eq:prop1}, we get
\[
\epsilon _{n+1}=-\frac{a_i R^i_1}{1+a_i \dot{x}^i_n}=-\frac{a_i \ddot{\beta }'^i(\xi ^i_n)}{2\left( 1+a_i \dot{x}^i_n\right) }\epsilon _n^2.
\]
Assuming that the acceleration of $\beta '$ is bounded and $|\epsilon _n|$ is sufficiently small (this can be achieved if $q_{\mathrm{part}}$ and $q_{\mathrm{geo}}$ are sufficiently close), the result holds.
\end{proof}

Once we obtain $\widetilde{q}_{\mathrm{s}}$, we have to estimate the vector $\widetilde{s}$ such that the geodesic starting from $p$ with initial tangent vector $\widetilde{s}$ arrives at $\widetilde{q}_{\mathrm{s}}$, i.e. $\widetilde{s}:=\exp_p^{-1}\widetilde{q}_{\mathrm{s}}$. Note that $\widetilde{s}$ might not be exactly orthogonal to $u$ because $q_{\mathrm{geo}}+T_{q_{\mathrm{geo}}}L_{p,u}$ does not coincide in general with $L_{p,u}$ (i.e. $\widetilde{q}_{\mathrm{s}}$ is an estimation of $q_{\mathrm{s}}$ and they do not coincide in general). Re-scaling $s_0$ in order to verify $\exp_p s_0=q_{\mathrm{geo}}$ (for convenience) and working in coordinates, we can make the linear estimation
\begin{equation}
\label{eq:qsi}
\widetilde{q}_{\mathrm{s}}^{\,\mu}\approx q_{\mathrm{geo}}^{\mu}+J_{\nu}^{\mu}\cdot (\widetilde{s}^{\,\nu}-s_0^{\nu}),
\end{equation}
where $J^{\mu}_{\nu}=\partial _{\nu} \exp _p ^{\mu} s_0$ was previously computed. Hence, solving the linear system given by \eqref{eq:qsi} we have
\begin{equation}
\label{eq:v1s}
s_1^{\nu}:=s_0^{\nu}+ (J^{-1})_{\mu}^{\nu}\cdot \left( \widetilde{q}_{\mathrm{s}}^{\, \mu}-q_{\mathrm{geo}}^{\mu}\right) \approx \widetilde{s} ^{\, \nu} ,
\end{equation}
where $(J^{-1})_{\mu}^{\nu}$ are the coefficients of the inverse of the Jacobian matrix $J_{\nu}^{\mu}$.

Finally, since $s_1$ might not be exactly orthogonal to $u$ (as it happens with $\widetilde{s}$), we have to redefine $s_1$ projecting it onto $u^{\bot}$, obtaining in this way an estimation of the desired vector $s$.

\begin{remark}
\label{rem:soa}
A higher order approximation based on the Taylor expansion of $\exp _p$ can be applied instead of the linear estimation \eqref{eq:qsi}. For example, the following quadratic approximation
\begin{equation}
\label{eq:qsi2}
\widetilde{q}_{\mathrm{s}}^{\,\mu}\approx q_{\mathrm{geo}}^{\mu}+\partial _{\nu} \exp _p ^{\mu} s_0\cdot (\widetilde{s}^{\,\nu}-s_0^{\nu})+\frac{1}{2}\partial _{\nu \alpha}\exp _p ^{\mu} s_0\cdot (\widetilde{s}^{\,\nu}-s_0^{\nu})\cdot (\widetilde{s}^{\,\alpha}-s_0^{\alpha}).
\end{equation}
In this case, the coordinates of $\widetilde{s}$ can be solved from the system of quadratic equations \eqref{eq:qsi2}, obtaining an expression equivalent to \eqref{eq:v1s} but involving a second order approximation. If there is more than one solution for $\widetilde{s}$, we have to take the solution such that $\exp _p \widetilde{s}$ is ``closest'' to $\widetilde{q}_{\mathrm{s}}$, e.g. that minimizes $\sum _{\mu=1}^4 |\exp _p^{\mu}{\widetilde{s}}-\widetilde{q}_{\mathrm{s}}^{\, \mu }|$. However, in practice, usually it would suffice to take the closest (in a coordinate sense) solution to $s_0$.
\end{remark}

Let $\psi _1$ be the geodesic with initial direction $s_1$. Probably, this new geodesic $\psi _1$ does not intersect the test particle $\beta '$ because we have made estimations, but it should be closer to $\beta '$ than the initial geodesic $\psi _0$ (see Remark \ref{rem:conv} for a discussion about convergence). So, applying this method again to the new geodesic, we get other values for $q_{\mathrm{geo}}$, $q_{\mathrm{part}}$, and $\widetilde{q}_{\mathrm{s}}$, from which we obtain $s_2$ and a geodesic $\psi _2$ (with initial direction $s_2$) that should be closer to $\beta '$ than $\psi _1$. We must repeat this process, obtaining a sequence of geodesics $\psi _n$ (with initial directions $s_n$, see Figure \ref{vdr-sp}), until we reach the desired accuracy or we exceed a certain number of iterations.

Summing up, given an observer's event $p=\beta (t)$ with $4$-velocity $u$, the algorithm for finding the relative position $s$ of the test particle $\beta '$ is as follows:
\begin{enumerate}
\item Choose an initial vector $s_0\in u^{\bot}$ such that the geodesic $\psi _0$ starting from $p$ with initial direction $s_0$ is sufficiently close to the test particle $\beta '$ (in the sense of Section \ref{sec:nearness}).

    Set $n=0$.
\item Find $q_{\mathrm{geo}}\in \psi _n$ and $q_{\mathrm{part}}\in \beta '$ with the same coordinate time that minimizes the distance $\Vert q_{\mathrm{part}}-q_{\mathrm{geo}}\Vert $ (see Remark \ref{rem3.1}).

    For $n\geq 1$ we must check that this distance is lesser than the corresponding distance computed for $n-1$; if it does not hold, then we should stop the algorithm because it could be a divergence symptom.
\item Find $\widetilde{q}_{\mathrm{s}}$ by means of a Newton-Raphson method: we find $\widehat{q}_{\mathrm{s}}$ by \eqref{eq:qsapprox}, and we redefine $q_{\mathrm{part}}$ as the event of $\beta '$ with the same coordinate time as $\widehat{q}_{\mathrm{s}}$; then, by \eqref{eq:qsapprox} again, we find another $\widehat{q}_{\mathrm{s}}$, repeating this process until $\widehat{q}_{\mathrm{s}}$ approximates $\widetilde{q}_{\mathrm{s}}$ with the desired accuracy. A quadratic convergence is assured by Proposition \ref{prop:conv1}.
\item Define $s_{n+1}^{\nu}:=s_n^{\nu}+ \left( J^{-1}\right)_{\mu}^{\nu}\cdot \left( \widetilde{q}_{\mathrm{s}}^{\, \mu}-q_{\mathrm{geo}}^{\mu}\right) $ (it is a linear approximation of $\widetilde{s}:=\exp _p^{-1}\widetilde{q}_{\mathrm{s}}$, see \eqref{eq:v1s}), where $s_n$ is re-scaled in order to hold $\exp _p s_n=q_{\mathrm{geo}}$. Alternatively, for a second order approximation, we can define $s_{n+1}$ by solving a system of quadratic equations (see Remark \ref{rem:soa}).
\item Redefine $s_{n+1}$ as its projection onto $u^{\bot}$: $s_{n+1}=s_{n+1}+g(u,s_{n+1})u$.
\item Set $n=n+1$. Repeat the process (steps 2, 3, 4, 5) with the geodesic $\psi _n$ starting from $p$ with initial direction $s_n$, until we reach the desired accuracy. Otherwise, we should stop the algorithm if we arrive at a predetermined maximum number of iterations.
\end{enumerate}

If the desired accuracy has been achieved, we can apply this algorithm again for another event $\beta(t+\Delta t)$ of the observer. Then, the new initial vector $s_0$ should be chosen as the vector with the same coordinates as the corresponding final vector $s_n$ computed for the previous event $\beta(t)$. In this case, choosing a sufficiently small time step $\Delta t$ assures convergence (see Remark \ref{rem:conv}), and ``differentiability'' of $S$ (although we obtain a discrete vector field) in the case of non-convex normal neighborhoods (see Section \ref{nonconvex}).

\begin{remark}
\label{rem:conv}
With respect to the convergence, in practice, this algorithm has been tested in many spacetimes with different test particles and observers, and it numerically converges in all of them with quadratic order. This accords with the fact that, if we compute $\widetilde{q}_{\mathrm{s}}$ exactly (in the 3rd step) and define $s_{n+1}:=\widetilde{s}$ (in the 4th step), the algorithm has a Newton-Raphson structure.

But, in general, it is not assured that $\psi _{n+1}$ is closer to $\beta '$ than $\psi _n$. If this property does not hold, then convergence is not guaranteed. Determining theoretical conditions for assuring convergence in a general case is a hard open problem. For this purpose, Proposition \ref{prop:conv1} is the first step.

Nevertheless, if the algorithm converges at the observer's event $p=\beta (t)$, then, applying differentiability arguments, it is assured that there exists a sufficiently small time step $\Delta t>0$ such that the algorithm also converges at $\beta (t+\Delta t)$ taking the new initial vector $s_0$ as the vector with the same coordinates as the corresponding final vector $s_n$ computed for the previous event $p$.
\end{remark}

\subsection{Lightlike simultaneity}
\label{sec:3.3}

\begin{figure}[tbp]
\begin{center}
\includegraphics[width=0.75\textwidth]{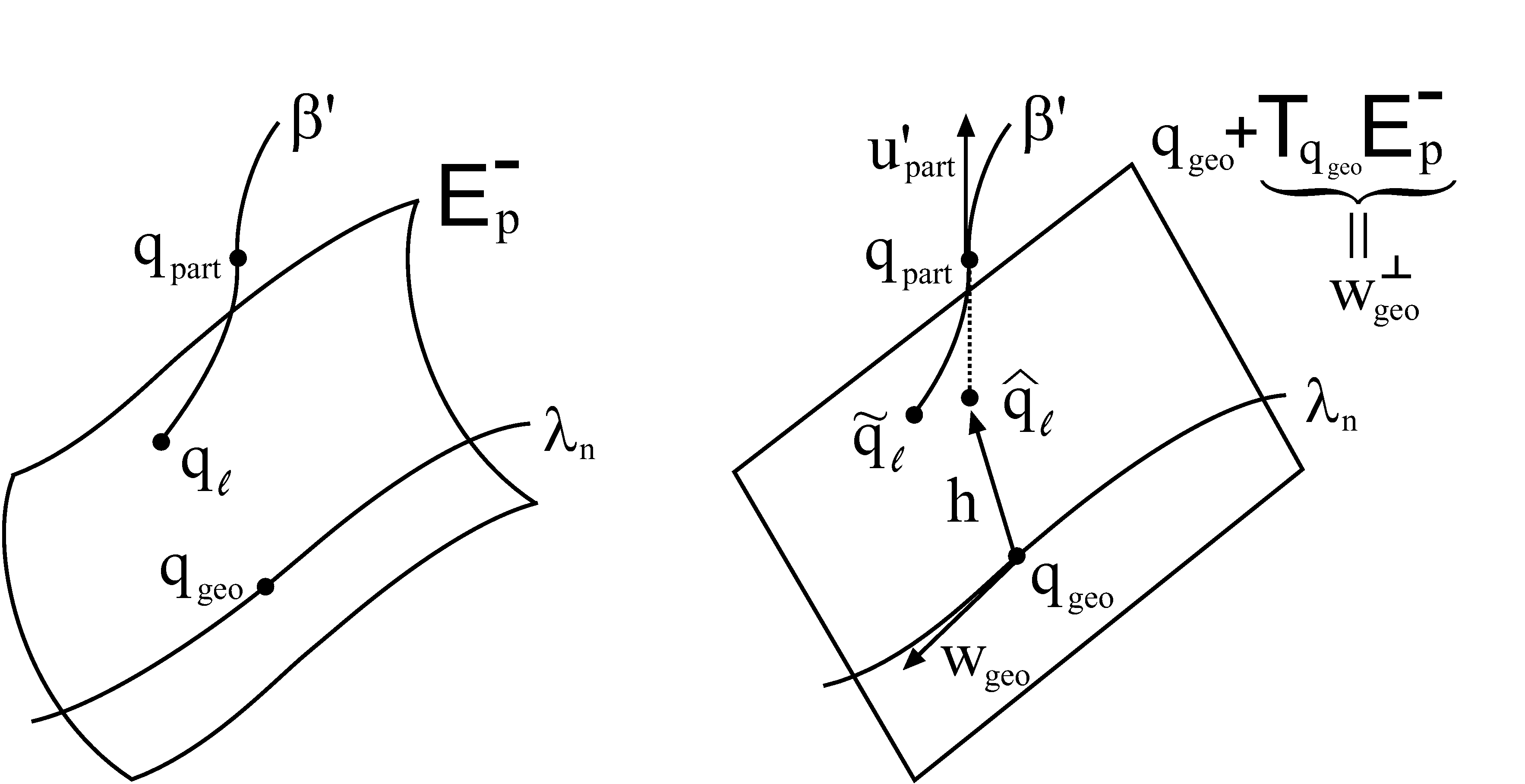}
\end{center}
\caption{Diagrams for $q_{\ell}$ (left) and the estimations of $q_{\ell}$ (right) when $q_{\mathrm{geo}}$ and $q_{\mathrm{part}}$ (that are events with the same coordinate time) are close. In this case, all the tangent spaces are identified by means of the coordinate system and provide an affine structure around $q_{\mathrm{geo}}$. Right: the past-pointing horismos submanifold $E^-_p$ is approximated by the affine hyperplane $q_{\mathrm{geo}}+T_{q_{\mathrm{geo}}}E^-_p$. Note that $T_{q_{\mathrm{geo}}}E^-_p=w_{\mathrm{geo}}^{\bot}$, where $w_{\mathrm{geo}}$ is the tangent vector of $\lambda _n$ at $q_{\mathrm{geo}}$.} \label{vdr-lg}
\end{figure}

Analogously, in the framework of lightlike simultaneity, the initial vector $w_0\in T_p\mathcal{M}$ must be lightlike and past-pointing, and the initial geodesic is named $\lambda _0$. Then, we estimate the intersection of the test particle with the affine hyperplane $q_{\mathrm{geo}}+T_{q_{\mathrm{geo}}}E^-_p$ using an expression analogous to \eqref{eq:qsapprox}:
\begin{equation}
\label{eq:qlapprox}
\widehat{q}_{\ell}:= q_{\mathrm{geo}}+h,
\end{equation}
where $h$ is the projection of $q_{\mathrm{part}}-q_{\mathrm{geo}}$ onto $T_{q_{\mathrm{geo}}}E^-_p$ parallel to $u'_{\mathrm{part}}$, with $u'_{\mathrm{part}}$ the $4$-velocity of $\beta '$ at $q_{\mathrm{part}}$ (see Figure \ref{vdr-lg} with $n=0$). Since $E^-_p$ is lightlike at $q$ (see Remark \ref{rem:lghor} below), we have that $h$ is always well-defined.

\begin{remark}
\label{rem:lghor}
Given $q\in E^-_p$, in \cite[Proposition 3]{Bol02} it is proved that $T_q E^-_p=\left( \tau _{pq} \exp _p^{-1}q\right) ^{\bot}$ and hence, the past-pointing horismos submanifold is always lightlike; in fact, we can write $T_q E^-_p=w^{\bot}$ where $w$ is the tangent vector of the light ray from $q$ to $p$. So,
\begin{equation}
\label{eq:qlapprox2}
\widehat{q}_{\ell} = q_{\mathrm{part}}-\frac{g(w_{\mathrm{geo}},q_{\mathrm{part}}-q_{\mathrm{geo}})}{g(w_{\mathrm{geo}},u'_{\mathrm{part}})}u'_{\mathrm{part}},
\end{equation}
where $w_{\mathrm{geo}}$ is the tangent vector of $\lambda _0$ at $q_{\mathrm{geo}}$. Note that $g(w_{\mathrm{geo}},u'_{\mathrm{part}})\neq 0$ because $w_{\mathrm{geo}}$ is lightlike and $u'_{\mathrm{part}}$ is timelike.
\end{remark}

Next, we have to find $\widetilde{q}_{\ell}$ applying a Newton-Raphson method analogously to the spacelike case. It can be proved analogously to Proposition \ref{prop:conv1} that this part of the algorithm has a quadratic order of convergence, but we only need the hypothesis on the bounded acceleration because $E^-_p$ is always lightlike (see Remark \ref{rem:lghor}).

After this, we have to estimate the vector $\widetilde{w}:=\exp_p^{-1}\widetilde{q}_{\ell}$. For this purpose, we re-scale $w_0$ in order to verify $\exp_p w_0=q_{\mathrm{geo}}$ and, working in coordinates analogously to \eqref{eq:v1s}, we obtain
\begin{equation}
\label{eq:v1l}
w_1^{\nu}:=w_0^{\nu}+ \left( J^{-1}\right)_{\mu}^{\nu}\cdot \left( \widetilde{q}_{\ell}^{\, \mu}-q_{\mathrm{geo}}^{\mu}\right) \approx \widetilde{w} ^{\, \nu},
\end{equation}
where, in this case, $\left( J^{-1}\right)_{\mu}^{\nu}$ are the coefficients of the inverse matrix of $J_{\nu}^{\mu}:=\partial _{\nu} \exp _p ^{\mu} (w_0)$. A second order approximation can be deduced analogously to Remark \ref{rem:soa} replacing, in expression \eqref{eq:qsi2}, $\widetilde{q}_{\mathrm{s}}$, $\widetilde{s}$, and $s_0$ with $\widetilde{q}_{\ell}$, $\widetilde{w}$, and $w_0$ respectively.

Finally, since $w_1$ might not be exactly lightlike (as it happens with $\widetilde{w}$), we have to redefine $w_1$ projecting it onto $C^-_p$, obtaining in this way an estimation of the desired vector $w$. For example, we can redefine it as the past-pointing lightlike vector with the same spatial coordinates as the original $w_1$ (i.e. a projection parallel to $\frac{\partial}{\partial t}|_p$).

The steps of the algorithm in the framework of lightlike simultaneity are analogous to those exposed at the end of Section \ref{sec:3.2} in the framework of spacelike simultaneity, replacing $s_n$, $u^{\bot}$, $\psi _n$, $\widehat{q}_{\mathrm{s}}$, and $\widetilde{q}_{\mathrm{s}}$ with $w_n$, $C^-_p$, $\lambda _n$, $\widehat{q}_{\ell}$, and $\widetilde{q}_{\ell}$ respectively. In the 3rd step, we can compute $\widehat{q}_{\ell}$ by means of expression \eqref{eq:qlapprox2}. Moreover, in the 5th step, we have to project $w_{n+1}$ onto $C^-_p$ as it is explained in the above paragraph.

\subsection{Non-convex normal neighborhoods}
\label{nonconvex}

Working in a convex normal neighborhood, there is no problem in the determination of the events $q_{\mathrm{geo}}$ (in the geodesic $\psi _n$ or $\lambda _n$) and $q_{\mathrm{part}}$ (in the test particle $\beta '$) introduced in the beginning of Section \ref{sec:3.2}, because they globally minimizes the distance between events of the geodesic and the test particle in surfaces of constant coordinate time (see Remark \ref{rem3.1}). But if we work in a non-convex normal neighborhood, then there could be different (or none) possibilities of relative position of the test particle with respect to the same event $p$ of the observer and, in this case, each possibility of relative position drives to a different local minimum of this distance.
Hence, for determining a ``suitable'' pair $q_{\mathrm{geo}},q_{\mathrm{part}}$ we have to search the local minimum that corresponds to a ``suitable'' relative position, according to the previously computed relative positions.

Let us explain it in more detail: suppose that we have previously applied the algorithm in $p_0:=\beta(t)$ and we have obtained a relative position $s_{p_0}$. Then, applying the algorithm in $p_1:=\beta(t+\Delta t)$, we say that a relative position $s_{p_1}$ is \textit{suitable} (according to $s_{p_0}$) if $s_{p_0}$ and $s_{p_1}$ are vectors of a (differentiable) vector field $S$ of relative positions. If the time step $\Delta t$ is sufficiently small, then it is assured that there exists at least one suitable $s_{p_1}$. If there are several relative positions, a suitable one must be close (in a coordinate sense) to $s_{p_0}$, and so it can be chosen as a relative position that, for example, minimizes the coordinate distance $\sum _{\nu =0}^4 |s_{p_1}^{\nu}-s_{p_0}^{\nu}|$. If there are several suitable relative positions, then any of them are valid. In this case, the output can be controlled adding some desired restrictions. For example, in the case of an equatorial circular geodesic test particle and an equatorial stationary observer in Schwarzschild spacetime (where there is gravitational lensing, see Example \ref{example3}), it could be desirable that all the involved geodesics ($\psi _n$ and $\lambda _n$) were also equatorial; we can impose this condition and so the relative positions $S$ and $S_{\mathrm{obs}}$ have zero $\theta $-component.

In practice, a \textit{suitable} pair $q_{\mathrm{geo}},q_{\mathrm{part}}$ for $p_1$ satisfies that the sum of the coordinate distances between them and the corresponding pair for $p_0$ is small, given a sufficiently small time step $\Delta t$.

Taking all this into account, the algorithm presented in Section \ref{sec3} returns one possibility of a discretized version of a relative position vector field $S$ or $S_{\mathrm{obs}}$, provided that there exists a relative position in the first event of the observer in which we apply the algorithm and we use a sufficiently small time step in the observer. From this output we obtain one discretized version of the (differentiable) vector fields $V_{\mathrm{kin}}$, $V_{\mathrm{Fermi}}$, $V_{\mathrm{spec}}$, or $V_{\mathrm{ast}}$.

\section{Examples}
\label{sec:examples}

The algorithm has been tested in several spacetimes with different observers and test particles, obtaining very good results in computation time and accuracy. We present here the most representative examples in Schwarzschild and Kerr spacetimes, using a test particle with equatorial geodesic orbit and stationary observers.

To sum up, given an observer and a test particle, our objective is to find $s$ (spacelike simultaneity) or $w$ (lightlike simultaneity) at a given event $p$ of the observer (see Figure \ref{diagram2}). To do this, we need an initial vector $s_0$ or $w_0$ that is supposed to be close (in a coordinate sense) to $s$ or $w$ respectively; but first, we are going to show by means of Examples \ref{example1} and \ref{example2} the rate of convergence of this method using an initial vector not necessarily close to the objective vector. Moreover, we are going to apply the second order approximation proposed in Remark \ref{rem:soa} and compare with the usual linear method.

\begin{example}
\label{example1}
The Schwarzschild metric in spherical coordinates $\left\{ t,r,\theta ,\varphi \right\} $ is given by the line element
\begin{equation}
\label{eq:sch}
ds^2 = -\left( 1-\frac{2m}{r}\right) \mathrm{d}t^2+\left( 1-\frac{2m}{r}\right) ^{-1}\mathrm{d}r^2+r^2\left( \mathrm{d}\theta ^2+\sin ^2\theta \mathrm{d}\varphi ^2 \right) ,
\end{equation}
where the parameter $m$ is interpreted as the mass of the gravitating object, $r>2m$ is the radial coordinate, and $0<\theta <\pi$. From now on we are going to suppose that the coordinates hold these restrictions and $m=1$.
In the framework of this coordinate system, a \textit{stationary observer} is an observer with constant spatial coordinates.
Note that stationary observers are not geodesic, but they are useful in the description and interpretation of the Schwarzschild spacetime.

Given a stationary observer at $r_0=8$, $\theta _0=\pi /2$, $\varphi _0=0$, and a test particle with equatorial circular geodesic orbit with $r_1=4$, $\theta _1=\pi /2$, $\varphi _1=\pi/2$ at $t=0$, we are going to check the algorithm for computing $s$ and $w$ at $p=(0,8,\pi/2,0)$ using an initial spacelike direction $s_0=(0,0,0,1)$ (orthogonal to the $4$-velocity of the observer at $p$), and using an initial past-pointing lightlike direction $w_0=(-\frac{2}{\sqrt{3}}r_0,0,0,1)$, respectively. Note that these initial directions are not too close of the objective vectors $s$ and $w$, but despite this, the method works well. Moreover, applying the second order approximation of Remark \ref{rem:soa} gives better results, but it doubles the computation time because we have to solve a system of quadratic equations (anyway, the computation time is a few seconds). The results are shown in Tables \ref{tab1}, \ref{tab1b}, \ref{tab2}, \ref{tab2b} (until reaching a relative error of order $10^{-6}$ or less) and Figure \ref{orb_schw_4-8}.

\begin{figure}[tbp]
\begin{center}
\includegraphics[width=0.8\textwidth]{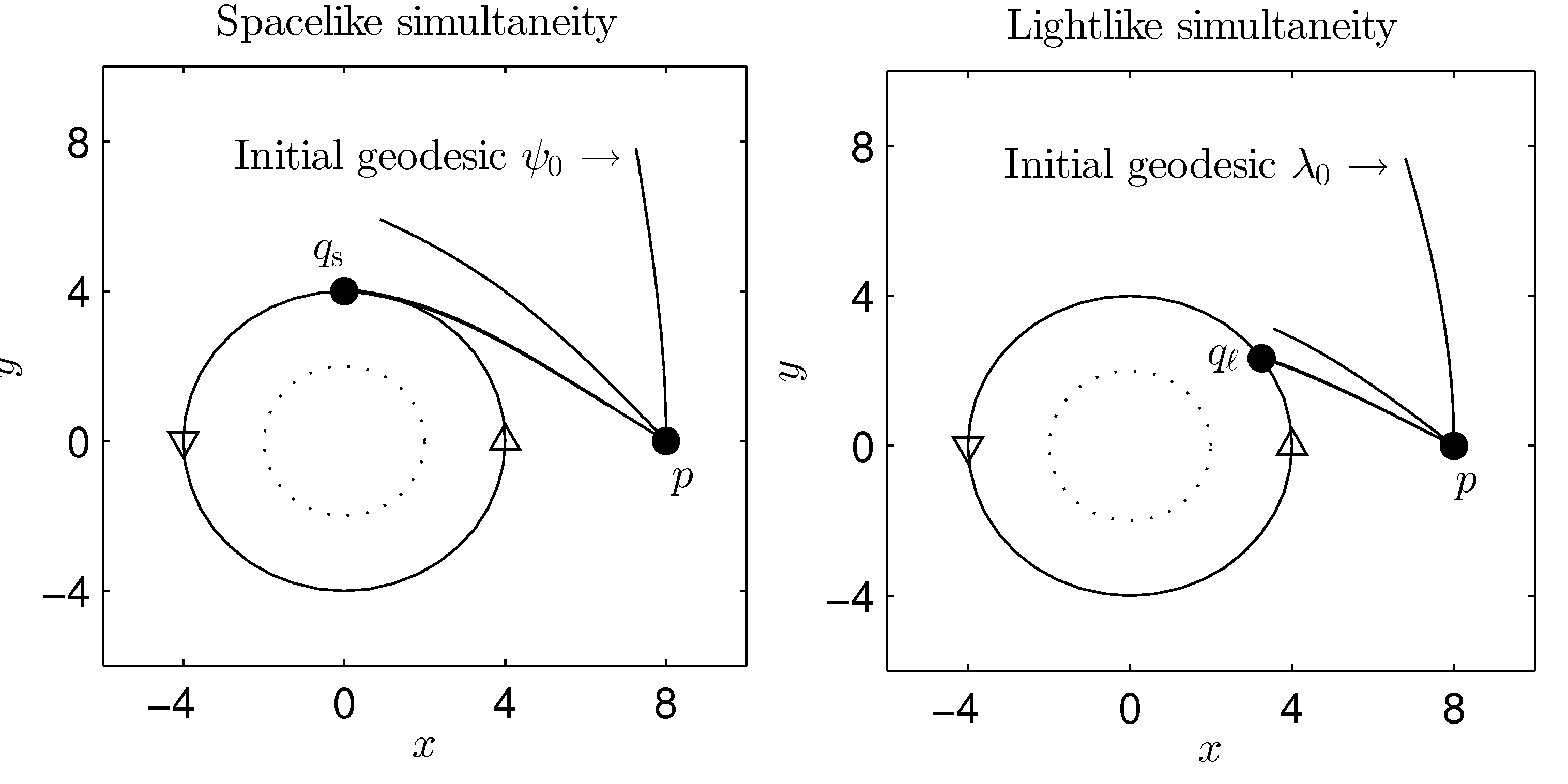}
\end{center}
\caption{Diagrams in the $xy$-plane of the elements involved in Example \ref{example1}. The stationary observer launches an initial geodesic from $p=(0,8,\pi/2,0)$ with initial direction $s_0=(0,0,0,1)$ (left) and $w_0=(-\frac{16}{\sqrt{3}},0,0,1)$ (right). It is shown how the successive iterations of the algorithm return geodesics that are getting closer to the desired intersection point with the test particle, $q_{\mathrm{s}}$ or $q_{\ell }$ (see Figure \ref{diagram2}).} \label{orb_schw_4-8}
\end{figure}

\begin{table}[htb]\footnotesize
\centering
\begin{tabular}{llll}
\hline
$n$ & $s_n$ & $\exp _p s_n$ & Rel. error \\
\hline
$0$ & $(0,0,0,1)$ & $(0,10.65556,\pi/2,0.821858)$ & $2.1$ \\
$1$ & $(0,-5.928809,0,0.835241)$ & $(0,5.981945,\pi/2,1.421563)$ & $5.9\cdot 10^{-1}$ \\
$2$ & $(0,-7.220423,0,0.604050)$ & $(0,4.050859,\pi/2,1.568339)$ & $1.4\cdot 10^{-2}$ \\
$3$ & $(0,-7.247942,0,0.597185)$ & $(0,4.000075,\pi/2,1.570793)$ & $2.1\cdot 10^{-5}$ \\
$4$ & $(0,-7.247982,0,0.597175)$ & $(0,3.9999998,\pi/2,1.570796)$ & $6.2\cdot 10^{-8}$ \\
\hline
\end{tabular}
\caption{Example \ref{example1}. Successive iterations of the algorithm for computing $s$. The event $\exp _p s_n$ approximates the intersection event $q_{\mathrm{s}}$. The relative error corresponds to the sum of the relative errors of each coordinate between $\exp _p s_n$ and $q_{\mathrm{s}}=(0,4,\pi/2,\pi/2)$ (see Figure \ref{orb_schw_4-8} left).}
\label{tab1}
\end{table}

\begin{table}[htb]\footnotesize
\centering
\begin{tabular}{llll}
\hline
$n$ & $s_n$ & $\exp _p s_n$ & Rel. error \\
\hline
$0$ & $(0,0,0,1)$ & $(0,10.65556,\pi/2,0.821858)$ & $2.1$ \\
$1$ & $(0, -6.940805,0,0.666507)$ & $(0,4.531059,\pi/2,1.538718)$ & $1.5\cdot 10^{-1}$ \\
$2$ & $(0,-7.247475,0,0.597303)$ & $(0,4.000944,\pi/2,1.570752)$ & $2.6\cdot 10^{-4}$ \\
$3$ & $(0,-7.247982,0,0.597175)$ & $(0,3.99999991,\pi/2,1.570796)$ & $2.9\cdot 10^{-8}$ \\
\hline
\end{tabular}
\caption{Example \ref{example1}. Analogous to Table \ref{tab1}, but using the second order approximation of Remark \ref{rem:soa}.}
\label{tab1b}
\end{table}

\begin{table}[htb]\footnotesize
\centering
\begin{tabular}{llll}
\hline
$n$ & $w_n$ & $\exp _p w_n$ & Rel. error \\
\hline
$0$ & $(-9.237604,0,0,1)$ & $(-8.991730,10.245645,\pi/2,0.844573)$ & $2.1$ \\
$1$ & $(-6.445192,-3.916650,0,0.408921)$ & $(-7.241158,4.720295,\pi/2,0.723016)$ & $3.9\cdot 10^{-1}$ \\
$2$ & $(-6.473183,-4.365314,0,0.306661)$ & $(-7.571348,4.020019,\pi/2,0.627855)$ & $1.4\cdot 10^{-2}$ \\
$3$ & $(-6.472659,-4.376997,0,0.303040)$ & $(-7.580813,4.000015,\pi/2,0.623213)$ & $1.4\cdot 10^{-5}$ \\
$4$ & $(-6.472658,-4.377011,0,0.303036)$ & $(-7.580825,3.999991,\pi/2,0.623207)$ & $2.2\cdot 10^{-6}$ \\
\hline
\end{tabular}
\caption{Example \ref{example1}. Successive iterations of the algorithm for computing $w$. The event $\exp _p w_n$ approximates the intersection event $q_{\ell }$. The relative error corresponds to the sum of the relative errors of each coordinate between $\exp _p w_n$ and $q_{\ell }$ (see Figure \ref{orb_schw_4-8} right). In this case, since the exact $t$ and $\varphi $ coordinates of $q_{\ell }$ are unknown, they have been assumed to be the $t$ and $\varphi$ coordinates of $\exp _p w_4$; the $r$ and $\theta $ coordinates of $q_{\ell }$ are $4$ and $\pi/2$ respectively.}
\label{tab2}
\end{table}

\begin{table}[htb]\footnotesize
\centering
\begin{tabular}{llll}
\hline
$n$ & $w_n$ & $\exp _p w_n$ & Rel. error \\
\hline
$0$ & $(-9.237604,0,0,1)$ & $(-8.991730,10.245645,\pi/2,0.844573)$ & $2.1$ \\
$1$ & $(-6.466102,-4.131156,0,0.366627)$ & $(-7.389968,4.398864,\pi/2,0.691871)$ & $2.3\cdot 10^{-1}$ \\
$2$ & $(-6.472884,-4.372292,0,0.304508)$ & $(-7.576998,4.008086,\pi/2,0.625106)$ & $5.6\cdot 10^{-3}$ \\
$3$ & $(-6.472658,-4.377008,0,0.303037)$ & $(-7.580822,3.999996,\pi/2,0.623209)$ & $1.0\cdot 10^{-6}$ \\
\hline
\end{tabular}
\caption{Example \ref{example1}. Analogous to Table \ref{tab2}, but using a second order approximation analogous to that of Remark \ref{rem:soa}.}
\label{tab2b}
\end{table}

\end{example}

\begin{example}
\label{example2}
The Kerr metric in Boyer-Lindquist coordinates $\left\{ t,r,\theta ,\varphi \right\} $ is given by the line element
\begin{equation}
\label{eq:kerr}
ds^2 = -\mathrm{d}t^2+\frac{\rho ^2}{\Delta }\mathrm{d}r^2+\rho ^2\mathrm{d}\theta ^2+\left( r^2+a^2\right) \sin ^2\theta \mathrm{d}\varphi ^2+\frac{2m}{\rho ^2}r\left(\mathrm{d}t-a\sin ^2\theta \mathrm{d}\varphi \right) ^2 ,
\end{equation}
where
\[
\Delta := r^2-2mr+a^2,\qquad \rho ^2 := r^2+a^2\cos ^2\theta .
\]
This metric describes the exterior gravitational field of a rotating mass $m$ with specific angular momentum $a=J/m$, where $J$ is the total angular momentum of the gravitational source (see \cite{Pug11} for restrictions on the coordinates). From now on we are going to suppose that $m=1$ and $a=1/2$.

Analogously to Example \ref{example1}, given a stationary observer at $r_0=8$, $\theta _0=\pi /2$, $\varphi _0=0$, and a test particle with equatorial circular geodesic orbit with $r_1=4$, $\theta _1=\pi /2$, $\varphi _1=\pi/2$ at $t=0$, we are going to check the algorithm for computing $s$ and $w$ at $p=(0,8,\pi/2,0)$ using an initial spacelike direction $s_0=(-1/3,0,0,1)$ (orthogonal to the $4$-velocity of the observer at $p$), and using an initial past-pointing lightlike direction $w_0=(-1/3-\sqrt{3091/6},0,0,1)$, respectively. The method works as well as in Example \ref{example1}; moreover, in the lightlike case, the linear estimation gives similar results to those applying the second order approximation proposed in Remark \ref{rem:soa}. The results are shown in Tables \ref{tab3}, \ref{tab3b}, \ref{tab4}, and \ref{tab4b} (until reaching a relative error of order $10^{-6}$ or less).

\begin{table}[htb]\footnotesize
\centering
\begin{tabular}{llll}
\hline
$n$ & $s_n$ & $\exp _p s_n$ & Rel. error \\
\hline
$0$ & $(-0.333333,0,0,1)$ & $(-0.244901,10.664776,\pi/2,0.822363)$ & $2.9$ \\
$1$ & $(-0.266769,-5.872871,0,0.800307)$ & $(-0.706191,5.752058,\pi/2,1.386445)$ & $8.4\cdot 10^{-1}$ \\
$2$ & $(-0.197485,-6.878255,0,0.592454)$ & $(-1.084179,4.026634,\pi/2,1.450633)$ & $1.5\cdot 10^{-2}$ \\
$3$ & $(-0.196280,-6.889558,0,0.588841)$ & $(-1.092323,4.000189,\pi/2,1.450052)$ & $1.0\cdot 10^{-4}$ \\
$4$ & $(-0.196272,-6.889639,0,0.588815)$ & $(-1.092382,3.9999995,\pi/2,1.450047)$ & $1.2\cdot 10^{-7}$ \\
\hline
\end{tabular}
\caption{Example \ref{example2}. Successive iterations of the algorithm for computing $s$. The event $\exp _p s_n$ approximates the intersection event $q_{\mathrm{s}}$. The relative error corresponds to the sum of the relative errors of each coordinate between $\exp _p s_n$ and $q_{\mathrm{s}}=(0,4,\pi/2,\pi/2)$ (see Figure \ref{orb_schw_4-8} left). In this case, since the exact $t$ and $\varphi $ coordinates of $q_{\mathrm{s}}$ are unknown, they have been assumed to be the $t$ and $\varphi$ coordinates of $\exp _p s_4$; the $r$ and $\theta $ coordinates of $q_{\mathrm{s}}$ are $4$ and $\pi/2$ respectively.}
\label{tab3}
\end{table}

\begin{table}[htb]\footnotesize
\centering
\begin{tabular}{llll}
\hline
$n$ & $s_n$ & $\exp _p s_n$ & Rel. error \\
\hline
$0$ & $(-0.333333,0,0,1)$ & $(-0.244901,10.664776,\pi/2,0.822363)$ & $2.9$ \\
$1$ & $(-0.215246,-6.683244,0,0.645738)$ & $(-0.968833,4.433499,\pi/2,1.449471)$ & $2.2\cdot 10^{-1}$ \\
$2$ & $(-0.196243,-6.889900,0,0.588730)$ & $(-1.092573,3.999382,\pi/2,1.450033)$ & $3.4\cdot 10^{-4}$ \\
$3$ & $(-0.196272,-6.889638,0,0.588815)$ & $(-1.092382,3.99999995,\pi/2,1.450047)$ & $1.1\cdot 10^{-8}$ \\
\hline
\end{tabular}
\caption{Example \ref{example2}. Analogous to Table \ref{tab3}, but using the second order approximation of Remark \ref{rem:soa}.}
\label{tab3b}
\end{table}

\begin{table}[htb]\footnotesize
\centering
\begin{tabular}{llll}
\hline
$n$ & $w_n$ & $\exp _p w_n$ & Rel. error \\
\hline
$0$ & $(-9.599460,0,0,1)$ & $(-9.288415,10.147354,\pi/2,0.854357)$ & $2.0$ \\
$1$ & $(-6.805715,-4.000178,0,0.432700)$ & $(-7.945461,4.553253,\pi/2,0.747437)$ & $2.9\cdot 10^{-1}$ \\
$2$ & $(-6.764920,-4.337490,0,0.356268)$ & $(-8.176669,4.005345,\pi/2,0.667962)$ & $3.4\cdot 10^{-3}$ \\
$3$ & $(-6.764068,-4.340567,0,0.355373)$ & $(-8.178553,3.999982,\pi/2,0.666767)$ & $4.4\cdot 10^{-6}$ \\
\hline
\end{tabular}
\caption{Example \ref{example2}. Successive iterations of the algorithm for computing $w$. The event $\exp _p w_n$ approximates the intersection event $q_{\ell }$. The relative error corresponds to the sum of the relative errors of each coordinate between $\exp _p w_n$ and $q_{\ell }$ (see Figure \ref{orb_schw_4-8} right). In this case, since the exact $t$ and $\varphi $ coordinates of $q_{\ell }$ are unknown, they have been assumed to be the $t$ and $\varphi $ coordinates of $\exp _p w_3$; the $r$ and $\theta $ coordinates of $q_{\ell }$ are $4$ and $\pi/2$ respectively.}
\label{tab4}
\end{table}

\begin{table}[htb]\footnotesize
\centering
\begin{tabular}{llll}
\hline
$n$ & $w_n$ & $\exp _p w_n$ & Rel. error \\
\hline
$0$ & $(-9.599460,0,0,1)$ & $(-9.288415,10.147354,\pi/2,0.854357)$ & $2.0$ \\
$1$ & $(-6.747007,-4.393588,0,0.339049)$ & $(-8.209192,3.906149,\pi/2,0.643900)$ & $6.1\cdot 10^{-2}$ \\
$2$ & $(-6.764157,-4.340248,0,0.355466)$ & $(-8.178359,4.000538,\pi/2,0.666892)$ & $3.4\cdot 10^{-4}$ \\
$3$ & $(-6.764068,-4.340567,0,0.355373)$ & $(-8.178554,3.999982,\pi/2,0.666767)$ & $4.4\cdot 10^{-6}$ \\
\hline
\end{tabular}
\caption{Example \ref{example2}. Analogous to Table \ref{tab4}, but using a second order approximation analogous to that of Remark \ref{rem:soa}. In this case, the result for $n=3$ is very similar to the case of the standard algorithm (see Table \ref{tab4}).}
\label{tab4b}
\end{table}

\end{example}

Next, we are going to give some examples in Schwarzschild and Kerr spacetimes about computing the relative velocities along an observer. In the figures, we make ``retarded comparisons'' (see \cite{BK12}) of their moduli, i.e. we compare velocities that are measured at the same event of the test particle.
This is important if we want to make a fair comparison of velocities in the framework of spacelike simultaneity with velocities in the framework of lightlike simultaneity, because the event of the test particle at which the velocity is measured depends on the chosen simultaneity: $q_{\mathrm{s}}$ (kinematic and Fermi) or $q_{\ell}$ (spectroscopic and astrometric), see Figure \ref{diagram2}. So, we plot the modulus of a relative velocity as function of $t_{\mathrm{s}}$ (kinematic and Fermi) or $t_{\ell}$ (spectroscopic and astrometric), that are the time coordinates of $q_{\mathrm{s}}$ and $q_{\ell}$ respectively. All the numerical data has been computed with a relative error less than $10^{-6}$.

\begin{example}
\label{example3}
In the Schwarzschild metric \eqref{eq:sch}, let us consider the stationary observer and the test particle of Example \ref{example1}, but now we are going to suppose that it has $\varphi _1=0$ at $t=0$ (i.e. the observer, the test particle and the singularity $r=0$ are aligned at $t=0$). This problem has not been previously studied analytically due to its complexity.

Note that in this case, we do not work in a convex normal neighborhood and so there is not a unique kinematic, Fermi, spectroscopic and astrometric relative velocities, depending on different geodesics joining the observer and the test particle (i.e. different choices of $\psi $ in the spacelike case, or $\lambda $ in the lightlike case, see Figure \ref{diagram2}). For example, if the test particle is at spatial coordinates $r_1=4$, $\theta _1=\pi/2$, $\varphi _1=0$, the observer and it can be joined by geodesics, $\psi $ or $\lambda $, giving whole turns around the black hole or not. In Figure \ref{schw_4-8} there are represented the moduli of the corresponding relative velocities in the case of equatorial geodesics joining the observer and the test particle following the convention that $\varphi _1$ also indicates the number of turns (and their direction) around the black hole: for example, for $\varphi _1=0$, the geodesic goes directly from the observer to the test particle; on the other hand, for $\varphi _1=2\pi $ the geodesic gives an equatorial whole turn counter-clockwise around the black hole before arriving at the test particle (see Figure \ref{fig4b}).

\begin{figure}[tbp]
\begin{center}
\includegraphics[width=0.9\textwidth]{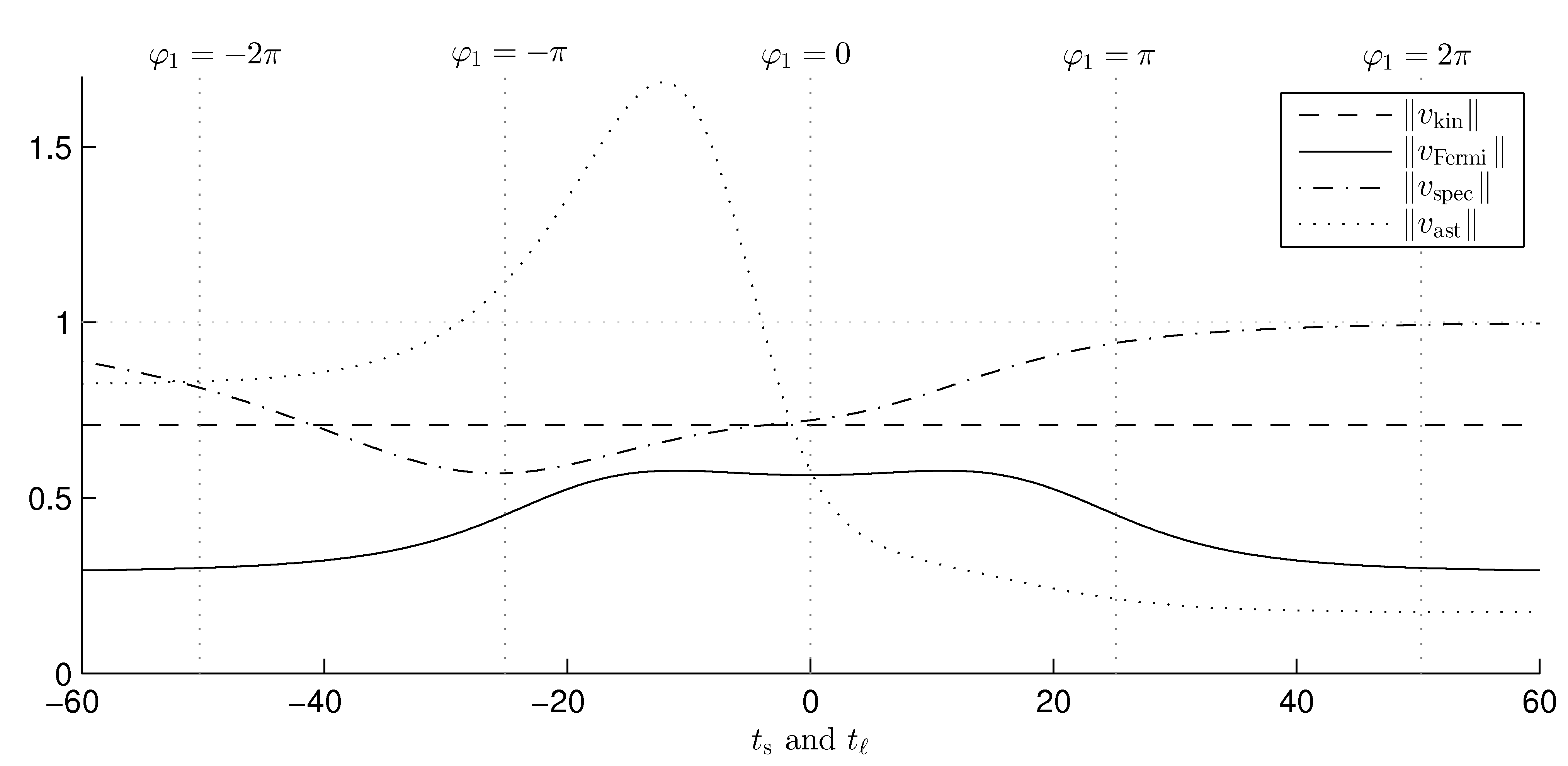}
\end{center}
\caption{Retarded comparison of the moduli of the kinematic, Fermi, spectroscopic and astrometric relative velocities of a test particle with equatorial circular geodesic orbit with radius $r_1=4$, $\theta _1=\pi /2$ and $\varphi _1=0$ at $t=0$, with respect to a stationary observer at $r_0=8$, $\theta _0=\pi /2$ and $\varphi _0=0$, in the Schwarzschild metric with $m=1$. Since it is a retarded comparison, we take as abscissa $t_{\mathrm{s}}$ (for kinematic and Fermi) and $t_{\ell}$ (for spectroscopic and astrometric), that are the time coordinates of the events $q_{\mathrm{s}}$ and $q_{\ell}$ respectively (i.e. the events of the test particle at which the corresponding velocities are measured, see Figure \ref{diagram2}). The vertical lines correspond to different values of $\varphi _1$ (the $\varphi $-coordinate of $q_{\mathrm{s}}$ or $q_{\ell}$), and the geodesics joining the observer and the test particle have been restricted to be equatorial.}
\label{schw_4-8}
\end{figure}

\begin{figure}[tbp]
\begin{center}
\includegraphics[width=1\textwidth]{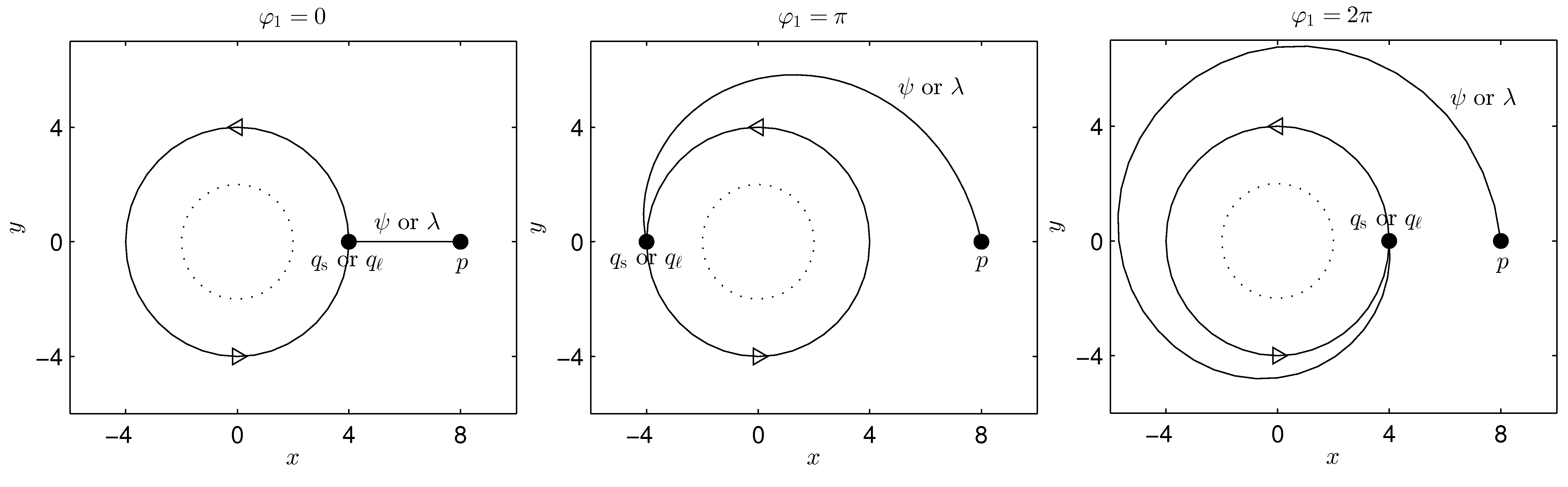}
\end{center}
\caption{Diagrams in the $xy$-plane of different equatorial geodesics joining the test particle and the observer at $p$ in Example \ref{example3}. The parameter $\varphi _1$ represents the $\varphi $-coordinate of $q_{\mathrm{s}}$ or $q_{\ell}$, and indicates the number of turns of the geodesic around the black hole.}
\label{fig4b}
\end{figure}

At $t=0$, we have taken initial vectors $s_0=(0,-1,0,0)$ and $w_0=(-4/3,-1,0,0)$ (for being a past-pointing lightlike vector at $p=(0,8,\pi /2,0)$), but this is not important because in this case the algorithm converges quickly also for non-nearby initial vectors and so it is not necessary to apply the second order approximation of Remark \ref{rem:soa}. The computations have been done using a time step (in the observer) of $\Delta t=0.25$, and the number of iterations needed at each time (for reaching the desired relative error $10^{-6}$) is at most $n=3$. So, the linear algorithm works very well in this case.

Considering stationary observers with different radial coordinate $r_0=4$ and $r_0=3$ (see Figures \ref{schw_4-4} and \ref{schw_4-3} respectively), we observe that $\|v_{\mathrm{kin}}\|$ remains constant and equal to $\sqrt{1/2}$. This numerical result has motivated a work \cite{Bol12b} where this property is theoretically proved in general: in the Schwarzschild metric, the modulus of the kinematic relative velocity of a test particle with circular geodesic orbit at radius $r_1>3m$ with respect to any stationary observer is constant and equal to $\sqrt{\frac{m}{r_1-2m}}$.

\begin{figure}[tbp]
\begin{center}
\includegraphics[width=0.9\textwidth]{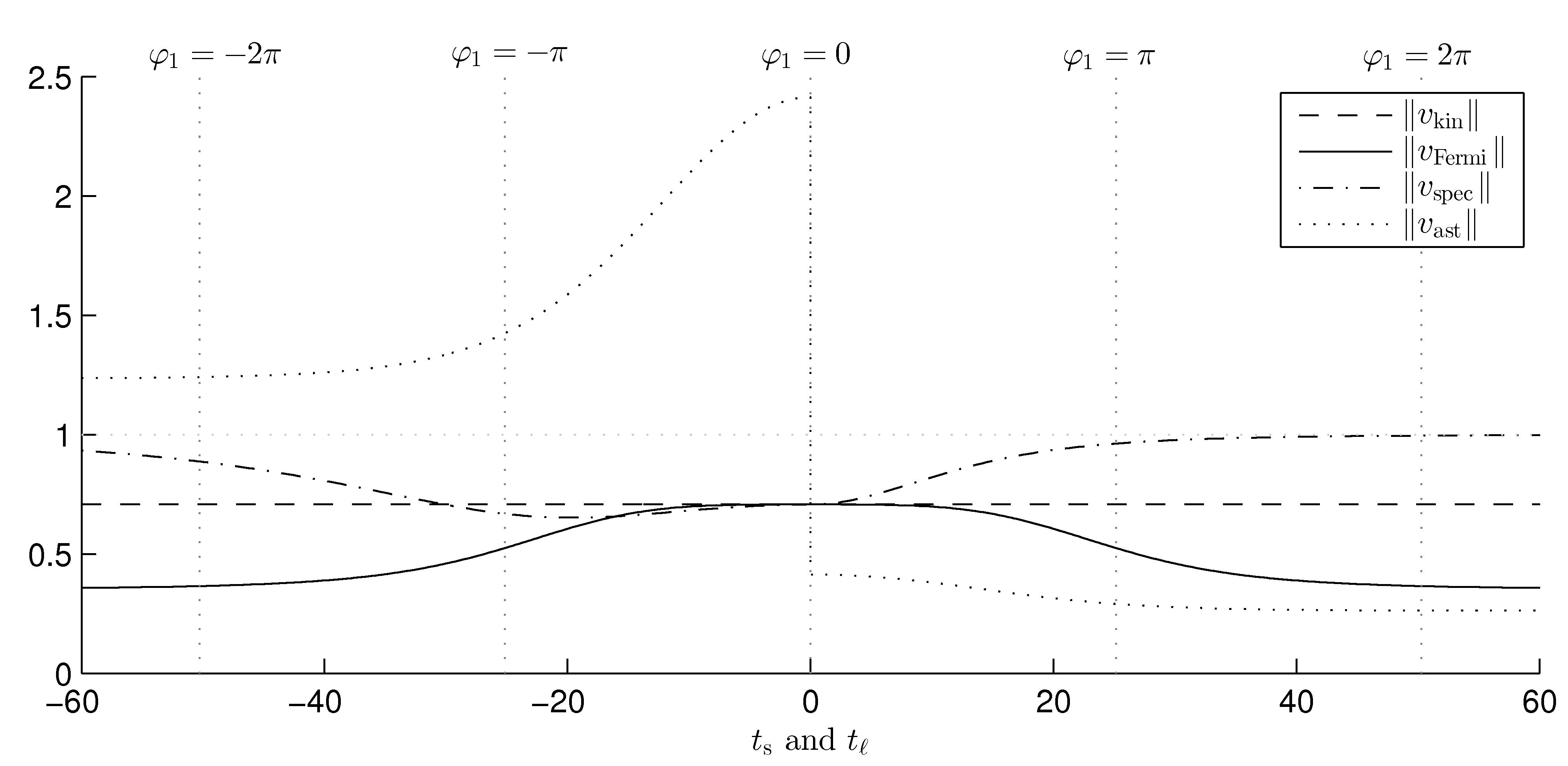}
\end{center}
\caption{Analogous to Figure \ref{schw_4-8}, but taking $r_0=4$.}
\label{schw_4-4}
\end{figure}

\begin{figure}[tbp]
\begin{center}
\includegraphics[width=0.9\textwidth]{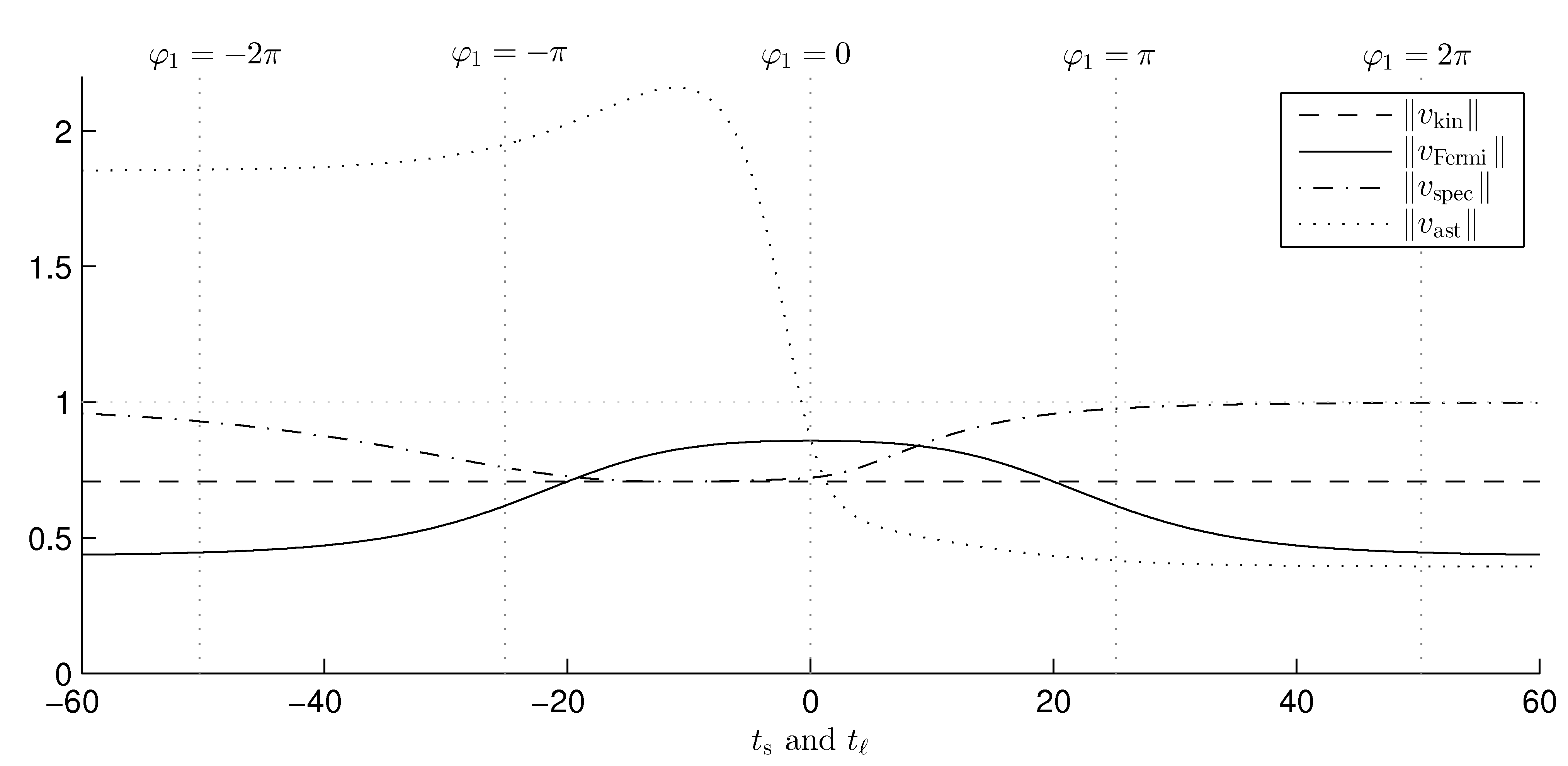}
\end{center}
\caption{Analogous to Figure \ref{schw_4-8}, but taking $r_0=3$.}
\label{schw_4-3}
\end{figure}

Moreover, it can be checked that $\|v_{\mathrm{spec}}\|$ tends to $1$ when $t_{\ell}\rightarrow \pm \infty $, and using expression \eqref{dopplergen}, we can compute the corresponding frequency shift, as it is seen in Figure \ref{shift_schw_4-3-4-8}.

\begin{figure}[tbp]
\begin{center}
\includegraphics[width=0.65\textwidth]{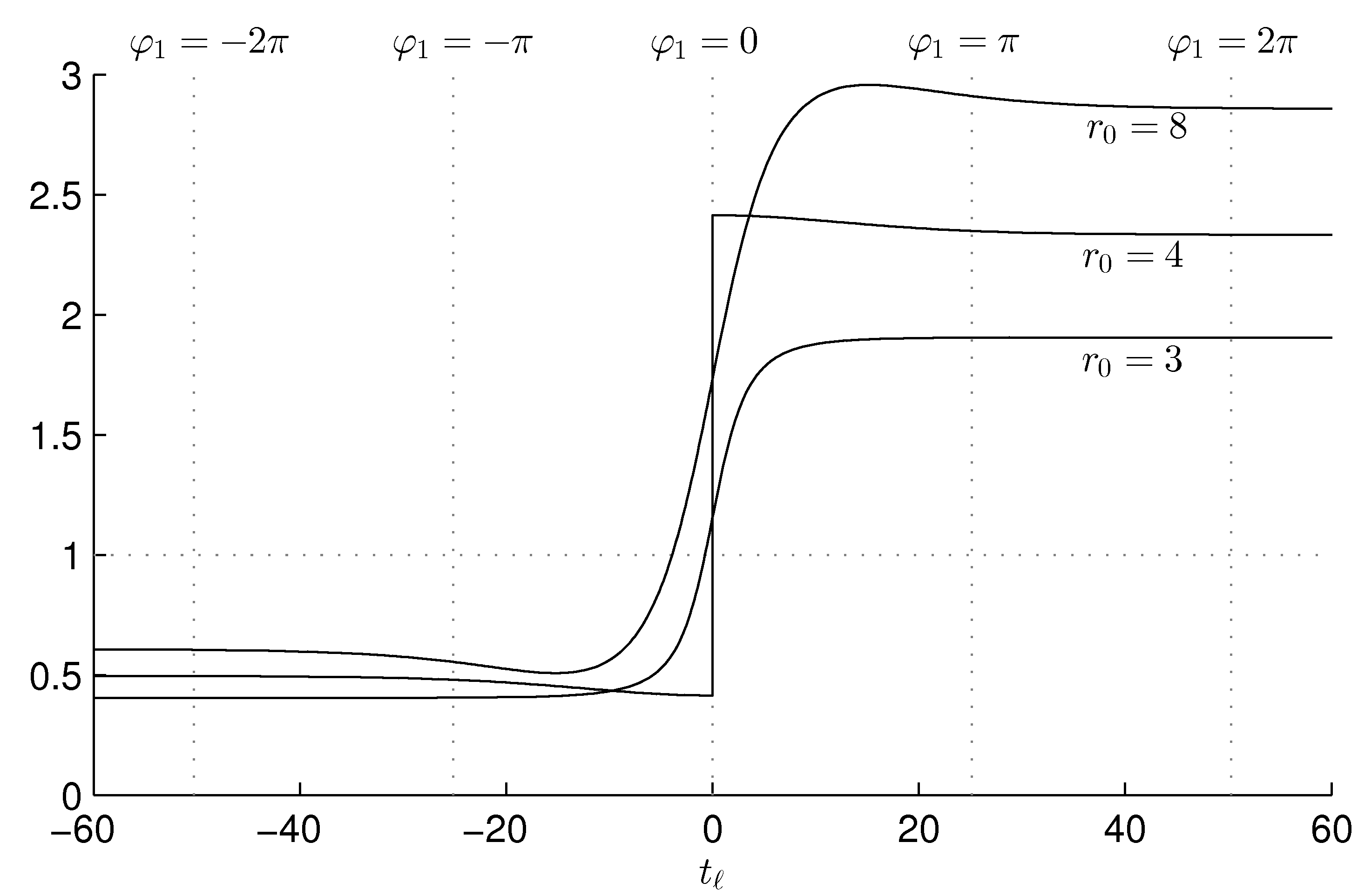}
\end{center}
\caption{Frequency shifts corresponding with the spectroscopic relative velocities computed in Example \ref{example3} (see Figures \ref{schw_4-8}, \ref{schw_4-4} and \ref{schw_4-3}).} \label{shift_schw_4-3-4-8}
\end{figure}

\end{example}

\begin{example}
\label{example4}
In the Kerr metric \eqref{eq:kerr}, let us consider the stationary observer and the test particle of Example \ref{example2}, but supposing that it has $\varphi _1=0$ at $t=0$, as in Example \ref{example3}. Analogously to this example, in Figure \ref{kerr_4-8} there are represented the moduli of the corresponding relative velocities. Comparing with the analogous problem in the Schwarzschild spacetime studied in Example \ref{example3} and Figure \ref{schw_4-8}, it can be observed that $\|v_{\mathrm{kin}}\|$ does not remain constant. Moreover, it also draws attention to the fact that $\|v_{\mathrm{spec}}\|$ does not tend to $1$ when $t_{\ell}\rightarrow -\infty$; in fact, it can be checked that it is decreasing and tends to a value $\approx 0.072$.

\begin{figure}[tbp]
\begin{center}
\includegraphics[width=0.9\textwidth]{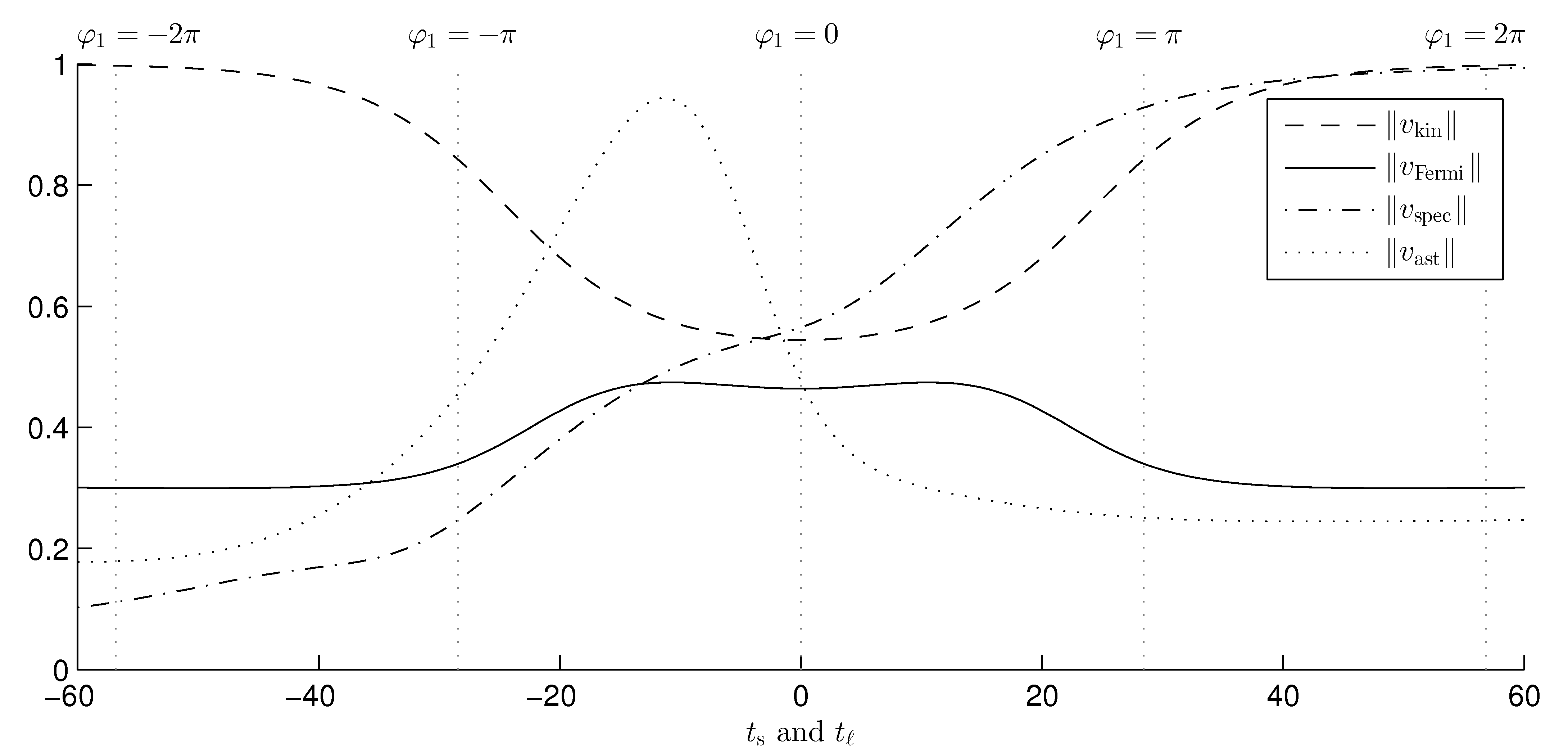}
\end{center}
\caption{Retarded comparison of the moduli of the kinematic, Fermi, spectroscopic and astrometric relative velocities of a test particle with equatorial circular geodesic orbit with radius $r_1=4$, $\theta _1=\pi /2$ and $\varphi _1=0$ at $t=0$, with respect to a stationary observer at $r_0=8$, $\theta _0=\pi /2$, $\varphi _0=0$, in the Kerr metric with $m=1$, $a=0.5$. The vertical lines correspond to different values of $\varphi _1$ (the $\varphi $-coordinate of the test particle at $q_{\mathrm{s}}$ or $q_{\ell}$), and the geodesics joining the observer and the test particle have been restricted to be equatorial.} \label{kerr_4-8}
\end{figure}

Finally, in Figure \ref{shift_schw-kerr} it is shown the frequency shift of the test particle with respect to the observer, compared with the frequency shift of the analogous problem in the Schwarzschild spacetime. Of note is the fact that, in the Kerr spacetime, the shift is greater than $1$ for a sufficiently negative $\varphi _1$; in fact, it tends to $1.053$ approximately when $\varphi _1\rightarrow -\infty $. Hence, in this case, an approaching\footnote{Taking into account the \textit{affine distance}, also known as \textit{lightlike distance}, defined as $\|s_{\mathrm{obs}}\|$ (see \cite{Bol05}).} test particle has redshift instead of blueshift.

\begin{figure}[tbp]
\begin{center}
\includegraphics[width=0.8\textwidth]{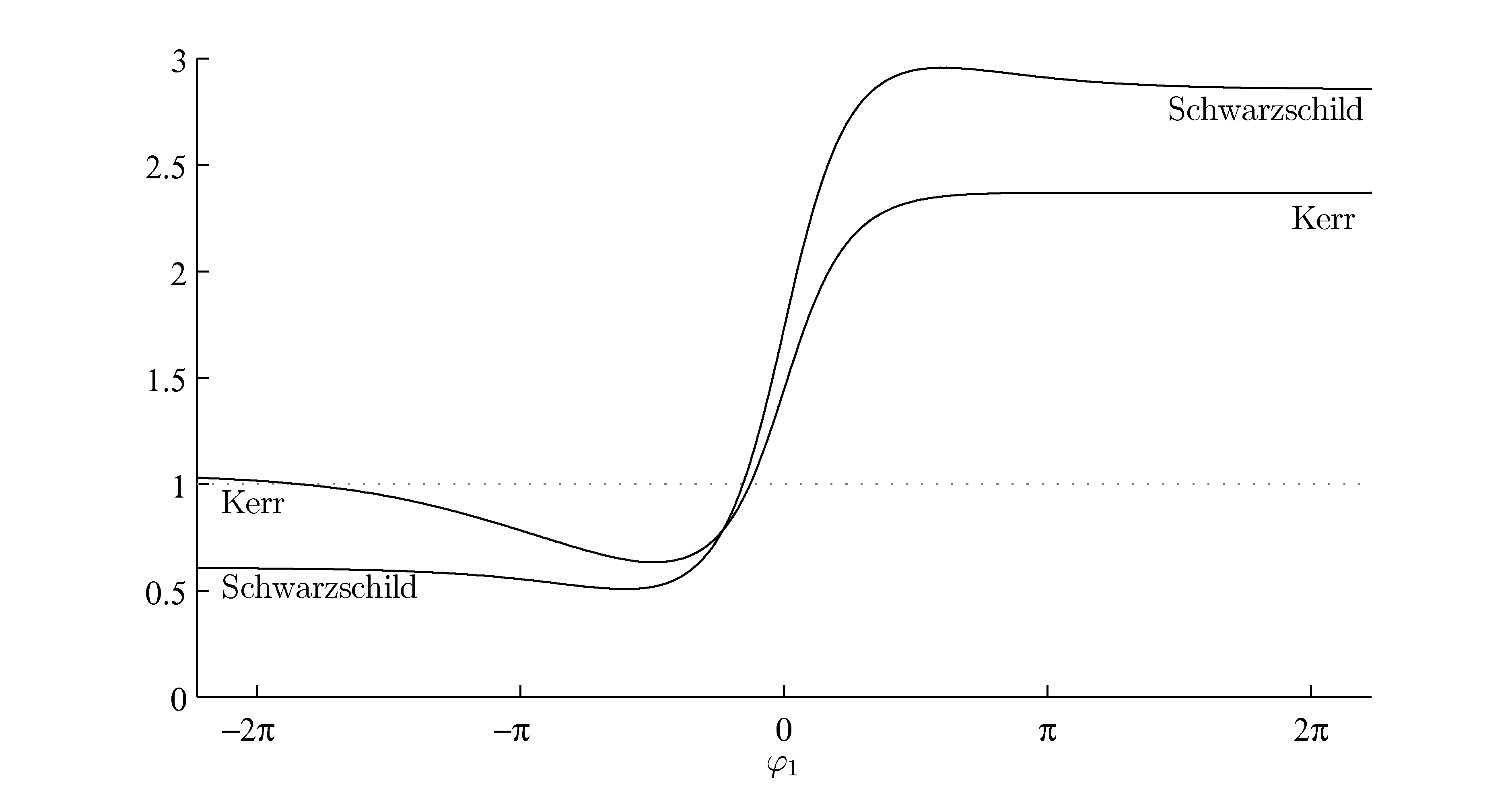}
\end{center}
\caption{Frequency shifts corresponding with the spectroscopic relative velocities computed in Example \ref{example4} (see Figure \ref{kerr_4-8}) and in the analogous problem in the Schwarzschild spacetime (see Figure \ref{schw_4-8} and Figure \ref{shift_schw_4-3-4-8} with $r_0=8$). The frequency shifts are plotted as functions of $\varphi _1$ (the $\varphi $-coordinate of the test particle at $q_{\ell}$) in order to make a fair comparison.}
\label{shift_schw-kerr}
\end{figure}

\end{example}

\section{Final remarks}

First, we have generalized the concepts of the relative velocities introduced in \cite{Bol07} to non-convex normal neighborhoods (see Section \ref{sec2.1}), focusing on what is minimally necessary to make sense of the corresponding definitions. As a result, we can now apply this theory of relative velocities to a wide range of scenarios, including those with gravitational lensing or caustics.

Then, we have developed an algorithm for computing relative velocities of a test particle with respect to an observer, based on finding the Fermi and observational coordinates of the test particle; hence, this method can be applied only for this purpose, allowing the fast computation of Fermi and observational coordinates with high accuracy.

With respect to Fermi coordinates, the objective of the paper \cite{KC08} is also to find the Fermi coordinates of an event in a general spacetime, calculating the general transformation formulas from arbitrary coordinates to Fermi coordinates, and vice versa. But the methods are not the same:
\begin{itemize}
\item In \cite{KC08}, the Fermi coordinates are given in the form of Taylor expansions. Hence, for increasing accuracy, we have to add terms of higher order, and it could be very complex to achieve a high accuracy for distant events.
\item In this work, we use an iterative algorithm. So, for increasing accuracy, we have to make more iterations (provided there is convergence, see Remark \ref{rem:conv}), and this does not imply additional difficulty.
\end{itemize}
Moreover, the algorithm proposed in this work is also valid in non-convex normal neighborhoods where there is not uniqueness of Fermi coordinates (see Section \ref{nonconvex}): depending on the initial vector $s_0$ you can get different relative positions $s$ of the same test particle at a given observer's time. Nevertheless, the method given in \cite{KC08} is a powerful tool in the theoretical study of Fermi coordinates, while the algorithm introduced here is designed to be implemented on computers. So, they are complementary and, for example, we can use the Taylor expansions given in \cite{KC08} for choosing the initial vector $s_0$ in the first execution. Moreover, if the algorithm does not converge (see Remark \ref{rem:conv}), then the Taylor expansions are a valuable alternative.

With respect to observational coordinates, they are used for the study of gravitational lensing or frequency shifts (e.g. the Pioneer anomaly), and they describe how we observe the universe. In this field, the numerical relativity is becoming more and more important due to the existence of complex models, such as dark matter models or multiple star systems.

\appendix

\section{Alternative computation of Fermi and astrometric relative velocities}
\label{sec:alter}

In \cite[Proposition 3.3]{Bol12} it is proved that, working in a convex normal neighborhood,  $v_{\mathrm{Fermi}}$ and $v_{\mathrm{ast}}$ can be computed in terms of $p$, $q_{\mathrm{s}}$, $q_{\ell }$, $u$, $\left( \nabla _U U\right) _p$, $u'_{\mathrm{s}}$, $u'_{\ell}$, $s$, $s_{\mathrm{obs}}$, and hence we do not need to know $S$ or $S_{\mathrm{obs}}$ around $p$, contrary to what is expected from \eqref{vfermi} and \eqref{vast}. This computation is done using a coordinate system $\left( x^0,x^1,x^2,x^3\right) $, obtaining
\begin{equation}
\label{vfermi2}
v_{\mathrm{Fermi}} = a_1 + \dot{\tau}'a_2 + a_3 - g\left( s,\left( \nabla _U U\right) _p\right) u,
\end{equation}
where the coordinates of vectors $a_1$, $a_2$, $a_3$ are given by
\[
a_1^{\mu} := \partial _{\nu} f_{\mathrm{s}}^{\mu} (p) u^{\nu}\qquad ;\qquad
a_2^{\mu} := \partial _{\nu} \log _p^{\mu} (q_{\mathrm{s}}) {u'_{\mathrm{s}}}^{\nu}\qquad ;\qquad
a_3^{\mu} := \Gamma ^{\mu}_{\nu \alpha}(p) u^{\nu} s^{\alpha},
\]
with
\begin{equation}
\label{fs}
f_{\mathrm{s}} := \log \left( \underline{\,\,\,\,} ,q_{\mathrm{s}}\right),
\end{equation}
where $\log (p,q)$ denotes $\log _p q$ (or equivalently $\exp_p^{-1}q$), and
\[
\dot{\tau}' = -\dfrac{g\left( s,\left( \nabla _U U\right) _p\right) + g\left( a_1+a_3,u\right) }{g\left( a_2,u\right) }.
\]
On the other hand
\begin{eqnarray}
\nonumber v_{\mathrm{ast}} & = & a_4 + \dot{\tau}'a_5 + \left( g\left( a_4 + \dot{\tau}'a_5,u\right) +g\left( \log _p q_{\ell}, \dot{u}\right) \right) u \\
\label{vast2} & & + g\left( \log _p q_{\ell},u\right) \dot{u} + a_6 - g\left( s_{\mathrm{obs}},\left( \nabla _U U\right) _p\right) u,
\end{eqnarray}
where the coordinates of vectors $a_4$, $a_5$, $a_6$ are given by
\[
a_4^{\mu} := \partial _{\nu} f_{\ell }^{\mu} (p) u^{\nu}\qquad ;\qquad
a_5^{\mu} := \partial _{\nu} \log _p^{\mu} (q_{\ell}) {u'_{\ell}}^{\nu}\qquad ;\qquad
a_6^{\mu} := \Gamma ^{\mu}_{\nu \alpha}(p) u^{\nu} s_{\mathrm{obs}}^{\alpha},
\]
with
\begin{equation}
\label{fl}
f_{\ell } := \log \left( \underline{\,\,\,\,} ,q_{{\ell}}\right) ,
\end{equation}
and
\[
\dot{\tau}' = \dfrac{g\left( s_{\mathrm{obs}},\left( \nabla _U U\right) _p\right) + g\left( a_4+a_6,u\right) +g\left( \log _p q_{\ell},g\left( \dot{u},u\right) u-\dot{u}\right) }{g\left( a_4,u\right) }.
\]
Note that $\dot{u}$ is given in terms of $p$, $u$ and $\left( \nabla _U U\right) _p$:
\[
\dot{u} = \left( \nabla _U U\right) _p-\Gamma ^{\mu}_{\nu \alpha}(p)u^{\nu}u^{\alpha}\left. \frac{\partial }{\partial x^{\mu}}\right| _p.
\]
Expressions \eqref{vfermi2} and \eqref{vast2} of $v_{\mathrm{Fermi}}$ and $v_{\mathrm{ast}}$ are not explicitly shown in \cite{Bol12}, but they can be deduced from the proof of Proposition 3.3.

Nevertheless, numerically it is difficult to compute the vectors $a_1$ and $a_4$ with high accuracy, concretely the derivatives of $f_{\mathrm{s}}$ \eqref{fs} and $f_{\ell}$ \eqref{fl}. For example, if we want to compute $\partial _{\mu}f_{\mathrm{s}}$, we are supposed to know $q_{\mathrm{s}}$ and the vector $\log _p q_{\mathrm{s}}$, i.e. the relative position $s$, that can be estimated by means of the algorithm exposed in Section \ref{sec:3.2}. Then, given a small $\epsilon >0$, we launch a geodesic from the event $p'$ with coordinates $p^{\nu}+\epsilon \delta ^{\nu}_{\mu}$ (i.e. the same coordinates as $p$ but the $\mu$-th coordinate is $p^{\mu}+\epsilon $) and initial tangent vector $s$ (actually, the vector in $T_{p'}\mathcal{M}$ with the same coordinates as $s$). Since $\epsilon $ is small, this geodesic is ``close''\footnote{We need another concept of ``nearness'' different from the one introduced in Section \ref{sec:nearness}, because it is not assured that the new geodesic intersects the leaf of constant coordinate time of $q_{\mathrm{s}}$; for example, a ``nearness'' based on the sum of the spatial coordinate distance with the temporal coordinate distance.} to $q_{\mathrm{s}}$, and so there is an event $q_{\mathrm{geo}}$ of the geodesic ``close'' to $q_{\mathrm{s}}$. So, assuming an affine structure, we have to parallel transport the vector $q_{\mathrm{s}}-q_{\mathrm{geo}}$ from $q_{\mathrm{geo}}$ to $p'$ along the geodesic, and add this vector to the current initial vector $s$, obtaining a new initial vector whose corresponding geodesic will be ``closer'' to $q_{\mathrm{s}}$. Repeating this process, we can estimate the initial vector of the geodesic passing through $q_{\mathrm{s}}$ with the desired accuracy and then, we can evaluate $\partial _{\mu}f_{\mathrm{s}}$ comparing this initial vector with $s$.

Analogously, for computing the derivatives of $f_{\ell}$, we are supposed to know $q_{\ell}$ and the vector $\log _p q_{\ell}$, whose projection onto $u^{\bot }$ is the observed relative position $s_{\mathrm{obs}}$, that can be estimated by means of the algorithm exposed in Section \ref{sec:3.3}.

Concluding, this method let us find $v_{\mathrm{Fermi}}$ and $v_{\mathrm{ast}}$ computing $S$ and $S_{\mathrm{obs}}$ only at $p$, but the original method based on definitions \eqref{vfermi} and \eqref{vast} (in which $S$ and $S_{\mathrm{obs}}$ are computed around $p$) is obviously faster and more accurate because it requires a far fewer number of operations. Moreover, if we do not work in a convex normal neighborhood, expressions \eqref{vfermi2} and \eqref{vast2} are not strictly valid because the vectors $a_1$, $a_2$, $a_4$, and $a_5$ are not well-defined in general.

\end{document}